\documentclass[11pt]{article}

\usepackage[top=1in, bottom=1in, left=1in, right=1in]{geometry}

\usepackage[utf8]{inputenc} 
\usepackage[T1]{fontenc}    
\usepackage{lmodern}
\usepackage{hyperref}       
\usepackage{url}            
\usepackage{booktabs}       
\usepackage{amsfonts}       
\usepackage{nicefrac}       
\usepackage{microtype}      
\usepackage{paralist}		
\usepackage{amsthm,amsmath,amssymb}
\usepackage{color,graphicx}
\usepackage{multicol}
\usepackage{multirow}
\usepackage{caption}
\usepackage{subcaption}
\usepackage[boxruled,linesnumbered]{algorithm2e}
\usepackage{algorithmic}
\usepackage{enumitem}
\usepackage{url}

\let\oldnl\nl
\newcommand{\nonl}{\renewcommand{\nl}{\let\nl\oldnl}}

\newtheorem{theorem}{Theorem}[section]

\newtheorem{claim}[theorem]{Claim}
\newtheorem{remark}[theorem]{Remark}

\newtheorem{lemma}[theorem]{Lemma}
\newtheorem{corollary}[theorem]{Corollary}

\newtheorem{definition}[theorem]{Definition}

\newtheorem{observation}[theorem]{Observation}
\newtheorem{example}[theorem]{Example}

\newcommand{\eps}{\varepsilon}
\renewcommand{\epsilon}{\varepsilon}

\newcommand{\eat}[1]{}

\newcommand{\R}{\mathbb{R}}

\newcommand{\poly}{\operatorname{poly}}
\newcommand{\diam}{\ensuremath{\mathsf{diam}}\xspace}

\newcommand{\OPT}{\ensuremath{\mathsf{OPT}}\xspace}

\newcommand{\FL}{Feldman-Langberg\xspace}

\newcommand{\calC}{{\mathcal{C}}}

\newcommand{\calG}{{\mathcal{G}}}

\newcommand{\calP}{{\mathcal{P}}}

\newcommand{\ProblemName}[1]{\textsc{#1}}
\newcommand{\kzC}{\ProblemName{$(k, z)$-Clustering}\xspace}
\newcommand{\kMedian}{\ProblemName{$k$-Median}\xspace}
\newcommand{\kMeans}{\ProblemName{$k$-Means}\xspace}
\newcommand{\kCenter}{\ProblemName{$k$-Center}\xspace}
\newcommand{\cost}{\ensuremath{\mathrm{cost}}\xspace}
\newcommand{\OkM}{\ProblemName{Ordered $k$-Median}\xspace}

\title{\bf Coresets for Clustering in Euclidean Spaces: Importance Sampling is Nearly Optimal}
\author{Lingxiao Huang \\ Yale University \and Nisheeth K. Vishnoi \\ Yale University}

\begin{document}

%
\maketitle

\begin{abstract}
	%
	%
	Given a collection of $n$ points in $\R^d$, the goal of the \kzC problem is to find a subset of $k$  ``centers'' that minimizes the sum of the $z$-th powers of the Euclidean distance of each point to the closest center.
	Special cases of the \kzC problem include the \kMedian and \kMeans problems.
	Our main result is a unified two-stage importance sampling framework that constructs an $\eps$-coreset for the \kzC problem.
	Compared to the results for \kzC in~\cite{feldman2011unified}, 
	our framework saves a $\eps^2 d$ factor in the coreset size.
	%
	{Compared to the results for \kzC in~\cite{sohler2018strong},
		our framework saves a $\poly(k)$ factor in the coreset size and avoids the $\exp(k/\eps)$ term in the construction time.}
	Specifically, our coreset for \kMedian ($z=1$) has size $\tilde{O}(\eps^{-4}k)$ which,
	when compared to the result in~\cite{sohler2018strong}, 
	saves a $k$ factor in the coreset size. 
	Our algorithmic results rely on a new dimension reduction technique
	%
	%
	%
	%
	that connects two well-known shape fitting problems: subspace approximation and clustering, and may be of  independent  interest.
	We also provide a size lower bound of $\Omega\left(k\cdot \min\left\{2^{z/20},d\right\}\right)$ for a $0.01$-coreset for \kzC, which has a linear dependence of size on $k$ and an exponential dependence on $z$ that matches our algorithmic results.
\end{abstract}

\thispagestyle{empty}

\newpage

\thispagestyle{empty}

\tableofcontents
\newpage

\setcounter{page}{1}

\section{Introduction}
\label{sec:introduction}

We study the problem of constructing coresets  for \kzC in Euclidean space $\R^d$ where $z\geq 1$ is constant.

\paragraph{\kzC in $\R^d$.} The input is a collection of $n$ points $X\subseteq \R^d$,
and the goal is to find a set $C \subseteq \R^d$ of $k$ points,
called \emph{center set}, that minimizes the objective function
\begin{equation} 
\label{eq:DefCost}
\cost_z(X, C) := \sum_{x \in X}{d^z(x, C)},
\end{equation}
where, throughout, $d^z$ denotes the Euclidean distance raised to power $z\ge 1$,
and 
\[
d(x, C):=\min\left\{d(x,c) = \sqrt{\sum_{i\in [d]} (x_i-c_i)^2}: c\in C\right\}.
\]
This formulation captures classical problems,
including the \kMedian problem (where $z = 1$) and  the \kMeans problem (where  $z=2$).
Moreover, this formulation can be generalized to weighted point sets in which each point $x\in X$ has a weight $u(x)$ and the goal is to compute a $k$-center set $C\subset \R^d$ that minimizes 
\[
\cost_z(X, C) := \sum_{x \in X}{u(x)\cdot d^z(x, C)}.
\]

\noindent
The \kzC problem is an essential tool in data analysis and is used in many application domains including approximation algorithms, unsupervised learning, and computational geometry~\cite{lloyd1982least,tan2006cluster,arthur2007k,coates2012learning}.
Due to its importance, several approximation algorithms for this clustering problem have been proposed~\cite{arya2004local,bandyapadhyay2016variants,mondal2018improved,cohen2019local}.

\paragraph{Coresets.} 
In recent years, a powerful data-reduction technique -- {coresets} --
has been used to find approximately optimal clustering in large datasets~\cite{harpeled2004on,feldman2011unified,feldman2013turning}.
Roughly speaking, a coreset is a ``compact'' summary of the data set, represented by a collection of weighted points, that approximates the clustering objective
for every possible choice of center set.
Let $\calC$ denote the collection of all ordered subsets (repetitions allowed) of $\R^d$ of size $k$ ($k$-center sets).
The coreset definition is as follows.

\begin{definition}[\bf{Coreset~\cite{langberg2010universal,feldman2011unified}}]
	\label{def:coreset}
	Given a collection $X\subseteq \R^d$ of $n$ weighted points and $\eps\in (0,1)$, an $\eps$-coreset for \kzC is a subset $S \subseteq \R^d$ with weights $w : S \rightarrow \R_{\geq 0}$ such that for any $k$-center set $C\in \calC$, the \kzC objective with respect to $C$ is $\eps$-approximately preserved, i.e.,
	\begin{equation*} 
	\sum_{x \in S}w(x) \cdot d^z(x, C)
	\in (1 \pm \eps) \cdot \cost_z(X, C).
	\end{equation*}
\end{definition}

\noindent
Coresets have been extensively studied in Euclidean spaces.
For \kzC in $\R^d$, Feldman and Langberg~\cite{feldman2011unified}
construct an $\eps$-coreset $O(\eps^{-2z} k d \log (k/\eps))$ based on an importance sampling framework.
Specifically, the dependence on dimension $d$ can be removed for \kMedian~\cite{sohler2018strong} and \kMeans~\cite{feldman2013turning,braverman2016new,sohler2018strong,becchetti2019oblivious}.

However, it is unknown whether we can similarly remove the dependence  on $d$ for general \kzC for arbitrary constant $z\geq 1$.
Also, the coreset for \kMedian~\cite{sohler2018strong} has a quadratic dependence of size on $k$, that does not match the size lower bound.
Moreover, the constructions of~\cite{sohler2018strong,becchetti2019oblivious} need a generalized notion of coreset instead of Definition~\ref{def:coreset}, which may increase the difficulty of applying existing clustering algorithms on coresets. 
Thus, it is an important problem to understand whether there exists a unified framework that constructs coresets satisfying Definition~\ref{def:coreset} for general \kzC, with a linear dependence on $k$, and no dependence on $d$.

\subsection{Our contributions}
\label{sec:contribution}

The main contribution of this paper is a unified framework that constructs $\eps$-coresets for \kzC of size $\tilde{O}(\min\left\{\eps^{-2z-2} k, 2^{2z}\eps^{-4} k^2\right\})$ and a nearly matching lower bound. 
We first propose a two-staged importance sampling framework that constructs coresets for \kzC (constant $z\geq 1$); summarized in the following theorem.
%

\begin{table}[t]
	\centering
	\caption{Summary of coreset size for \kzC.}
	\label{tab:result}
	
	\begin{tabular}{cc|c|c}
		\toprule
		& Reference & Objective & Coreset Size \\
		\midrule
		\multirow{10}{*}{Upper bounds} & \cite{feldman2011unified} & \kzC & $O\left(\eps^{-2z} k \textbf{\emph{d}} \log (k/\eps)\right)$ \\
		& \cite{varadarajan2012sensitivity} & \kzC & $O\left(2^{2z} \eps^{-2} k \textbf{\emph{d}} \log (k/\eps)\right)$ \\
		& \cite{sohler2018strong} & \kzC &  $\poly(k/\eps^z)$ \footnotemark[1]\\
		& This paper & \kzC & $O\left(\min\left\{\eps^{-2z-2}, 2^{2z} \eps^{-4} k \right\}\cdot k \log k \log(k/\eps)\right)$ \\
		\cline{2-4}
		& \cite{feldman2011unified} & \kMedian & $O\left(\eps^{-2} k \textbf{\emph{d}} \log k\right)$ \\
		& \cite{sohler2018strong} & \kMedian &  $O\left(\eps^{-4} k^{\mathbf{2}} \log k\right)$\\
		& This paper & \kMedian & $O\left(\eps^{-4}\cdot k\log k \log(k/\eps)\right)$ \\
		\cline{2-4}
		& \cite{braverman2016new} & \kMeans & $O\left(\eps^{-3} k^2 \log (k/\eps)\right)$ \\
		& \cite{becchetti2019oblivious} & \kMeans & $O\left(\eps^{-6} k \log^2 (k/\eps) \log (1/\eps)\right)$ \\
		& This paper & \kMeans & $O\left(\eps^{-6}\cdot k\log k \log(k/\eps)\right)$ \\
		\hline
		\multirow{2}{*}{Lower bounds} & \cite{braverman2019coresets} & \kMedian ($d=1$) & $\Omega\left(\eps^{-1/2} k\right)$ \\
		& This paper & \kzC &  $\Omega(2^{z/20} k)$ \\
		\bottomrule
	\end{tabular}

\end{table}

\footnotetext[1]{The paper did not present this result directly. But their approach can be easily generalized to \kzC.}

\begin{theorem}[\bf{Informal, see Theorem~\ref{thm:coreset}}]
	\label{thm:coreset_informal}
	There exists a randomized algorithm that, given a dataset $X$ of $n$ points in $\R^d$, $\eps\in (0,1/2)$, constant $z\geq 1$ and integer $k \geq 1$, constructs an $\eps$-coreset of size $\tilde{O}(\min\left\{\eps^{-2z-2} k, 2^{2z}\eps^{-4} k^2\right\})$ for \kzC, and runs in time $\tilde{O}(ndk)$.
\end{theorem}

\noindent
We compare our results with existing coreset results for \kzC in  Table~\ref{tab:result}.
This is the first result that constructs an $\eps$-coreset for \kzC whose size is independent of $d$ and near-linearly dependent of $k$.
Note that if $\eps, z$ are constants, this result saves a $d$ factor compared to prior results~\cite{feldman2011unified,varadarajan2012sensitivity} and matches the size lower bound.
Compared to the result in~\cite{sohler2018strong}, our coreset saves a $\poly(k)$ factor in the size and can be constructed in polynomial time -- avoids the expontential term ($\exp(k/\eps)$) and the dependence of $1/\eps$ ($n \poly(k/\eps)$) in the construction time of~\cite{sohler2018strong}.
Specifically, for \kMedian, our result saves a $k$ factor compared to~\cite{sohler2018strong}.
Our construction applies a unified two-staged importance sampling framework (Algorithm~\ref{alg:coreset}). 
Compared to existing approaches~\cite{sohler2018strong,becchetti2019oblivious} that require to apply projection methods, our construction is simple to implement. 
Also note that our coreset satisfies (the standard) Definition~\ref{def:coreset}, instead of the one that requires offsets as in recent results~\cite{sohler2018strong,becchetti2019oblivious}.
Consequently, we can directly combine existing clustering algorithms with our coresets to estimate \kzC objectives. 
It is an interesting open problem to investigate whether one-stage importance sampling could produce coresets with a comparable size.

\sloppy
We also extend Theorem~\ref{thm:coreset_informal} to $\ell_p$-metrics whose distance function is $d_p(x,y) = \left(\sum_{i\in [d]} |x_i-y_i|^p \right)^{1/p}$ ($x,y\in \R^d$) instead of the Euclidean distance in Equation~\eqref{eq:DefCost}; see the following corollary.

\begin{corollary}[\bf{Informal, see Corollary~\ref{cor:coreset_lp}}]
	\label{cor:coreset_informal}
	There exists a randomized algorithm that, given a dataset $X$ of $n$ points in $\R^d$, $1\leq p<2$, $\eps\in (0,1/2)$, constant $z\geq 1$ and integer $k \geq 1$, constructs an $\eps$-coreset of size $\tilde{O}(\min\left\{\eps^{-4z-2} k, 2^{4z}\eps^{-4} k^2\right\})$ for \kzC with $\ell_p$-metric, and runs in time $\tilde{O}(ndk)$.
\end{corollary}

\noindent
The main idea is that for $1\leq p < 2$, there exists an isometric embedding from $\ell_p$ to $\ell_2$ square~\cite{kahane1993some}. 
By this idea, we can reduce the problem of constructing an $\eps$-coreset for \kzC with $\ell_p$-metric (Definition~\ref{def:coreset_lp}) to constructing an $O(\eps)$-coreset for \ProblemName{$(k, 2z)$-Clustering} with $\ell_2$-metric (Definition~\ref{def:coreset}).
It is interesting to investigate whether the above corollary can be extended to all constant $p\geq 1$.

We also provide a matching size lower bound (Theorem~\ref{thm:lower}).

\begin{theorem}[\bf{Size lower bound}]
	\label{thm:lower}
	For every $z> 0$ and integers $d,k\geq 1$, there exists a point set $X$ in the Euclidean space $\R^d$ such that any $0.01$-coreset for \kzC over $X$ has size $\Omega\left(k\cdot \min\left\{2^{z/20},d\right\}\right)$.
\end{theorem}

\noindent
 To the best of our knowledge, this is the first result that shows that the coreset size for \kzC should exponentially depend on $z$.
However, tight dependence of size on the parameter $\eps$ is still unknown.

For the algorithmic result (Theorem~\ref{thm:coreset_informal}), our main technical contribution is a new dimension reduction technique which reduces the dimension of the space of $k$-center sets to $\poly(k/\eps)$. 
Towards this, we develop two new notions: representativeness and weak-coreset (Section~\ref{sec:definition}).
%
%
Given a collection $X$ in $\R^d$, we first divide all $k$-center sets into sub-groups by defining ``equivalence classes'' of $k$-center sets.
Our equivalence classes are induced by some subspace $\Gamma \subseteq \R^d$, in which each class is induced by  a $k$-center set in $\Gamma$ (Definition~\ref{def:equivalent}). 
Based on these equivalence classes, we define a  representativeness property (Definition~\ref{def:representativeness}), i.e., \kzC objectives over $X$ with respect to all $k$-center sets in an equivalence class are roughly the same.
We show that $X$ satisfies the representativeness property with respect to certain $\poly(k/\eps)$-dimensional subspace $\Gamma$ constructed by~\cite[Algorithm 1]{sohler2018strong} (Observation~\ref{ob:equivalent}).
Moreover, we present a sufficient condition for any weighted subset $S\subseteq X$ such that the representativeness property also holds for $S$, i.e.,  \kzC objectives over $S$ with respect to all $k$-center sets in an equivalence class are roughly the same (Observation~\ref{ob:equivalent_subset}).
The sufficient condition roughly requires that $S$ is a ``weak-coreset'' for $(k,z)$-subspace approximation over $X$ (Definition~\ref{def:weak_coreset}).
To satisfy this requirement, we only need the size of $S$ to be $\poly(k/\eps)$ (Theorem~\ref{thm:subspace}).
This enables us to only approximately preserve \kzC objectives in the low-dimensional subspace $\Gamma$ instead of $\R^d$, which leads to an $\eps$-coreset $D$ for \kzC of size $\poly(k/\eps)$. 

Compared to the \FL framework~\cite{feldman2011unified}, we successfully remove the dependence in coreset size of $d$ (Theorem~\ref{thm:upper}).
Moreover, by a well-known dimension reduction approach, called ``terminal embedding'' (Definition~\ref{def:embedding}), the dimension of $D$ can be further reduced to $O(\eps^{-2}\log(k/\eps))$. 
Thus, we can further reduce the coreset size to $\tilde{O}(\eps^{-2z-2} k)$ by applying an importance sampling framework over $D$ (Theorem~\ref{thm:reduction}).
Overall, our dimension reduction technique connects two well-known shape fitting problems: subspace approximation and clustering, and leads to a unified two-staged importance sampling framework for \kzC that removes the dependence of coreset size on $d$.
The geometric observations and notions, including equivalence classes and representativeness property, may be of independent  interest.

To establish a nearly tight size lower bound (Theorem~\ref{thm:lower}), we construct an instance $X$ of size $\Omega\left(k\cdot \min\left\{2^{z/20},d\right\}\right)$ such that for any point $x\in X$, there exists a $k$-center set $C_x$ which satisfies the condition that the clustering objective $d^z(x,C_x)\approx \cost_z(X,C_x)$.
To gain some intuition about the construction, if $k=1$ and $d\approx 2^{z/20}$ we let $X=\left\{e_1,-e_1,\ldots, e_d,-e_d\right\}$ and observe that $d^z(e_i,-e_i)=2^z \approx \cost_z(X,-e_i)$ for any $i\in [d]$.
Suppose $S$ is a $0.01$-coreset and $S\subseteq X$.
Then $S$ satisfies that $\cost_z(S,-e_i)\approx \cost_z(X,-e_i)$ for each $i\in [d]$.
Intuitively, each $e_i$ should be included in $S$ and, hence, $|S|= O(2^{z/20} k)$. 
The technical difficulty is that points in a coreset can come from $\R^d\setminus X$.
We show that even though $e_i$ may not be included in $S$, a  point close to $e_i$ must be included.
This results in a matching lower bound $|X|$.
%


\subsection{Other related works}
\label{sec:related}
Har-Peled and Mazumdar~\cite{harpeled2004on} constructed the first coreset for both \kMedian and \kMeans, however the size of their coresets depended exponentially on $d$.
Subsequently, Chen~\cite{chen2009on} improved dependence on the dimension to be polynomial for both \kMedian and \kMeans.
For \kMeans, Feldman et al.~\cite{feldman2013turning} designed coresets of size independent of $d$, which was improved by~\cite{braverman2016new} to be $\tilde{O}(\eps^{-3} k^2)$. 
Recently, the dependence on $k$ has been improved to near-linear by~\cite{becchetti2019oblivious}. 
For \kMedian, Sohler and Woodruff~\cite{sohler2018strong} show how to remove the dependence on $d$.
Recently, coreset for generalized clustering objective receives attention from the research community, for example, Braverman et al.~\cite{braverman2019coresets} obtain simultaneous coreset for \OkM.
Coresets for the fair version of \kMedian and \kMeans have also been investigated~\cite{schmidt2019fair,huang2019coresets}.
For another special case $z=\infty$, which is the \kCenter clustering, an $\eps$-coreset of size $O(\eps^{-d+1} k)$ can be constructed in near-linear time~\cite{agarwal2002exact,har2004clustering}.
This size has been proved to be tight for \kCenter [D. Feldman, private communication and~\cite{braverman2019coresets}].
For general \kzC (constant $z\geq 1$), Feldman and Langberg~\cite{feldman2011unified} construct an $\eps$-coreset of size $\tilde{O}(\eps^{-2z} kd)$, and recently this result has been generalized to doubling metrics~\cite{huang2018epsilon}.

For general metrics, an $\eps$-coreset for the \kzC problem of size $O(\eps^{-2z} k\log n)$ can be constructed in time $\tilde{O}(nk)$~\cite{feldman2011unified}, and
the $\log n$ factor is unavoidable~\cite{braverman2019treewidth}.
Specifically, for $k$-means clustering, Braverman et al.~\cite{braverman2016new} show a construction of size $O(\eps^{-2} k\log k\log n)$.

\section{Our notions of representativeness and weak-coreset}
\label{sec:definition}

In this section, we propose new definitions that are important for our dimension reduction technique.
Recall that $\calC$ denotes the collection of all ordered subsets (repetition allowed) of $\R^d$ of size $k$ ($k$-center sets).
We first define equivalence classes of $k$-center sets that partition $\calC$.

\paragraph{Equivalence classes of $k$-center sets and representativeness property.} 
Given a subspace $\Gamma \subsetneq \R^d$, we denote by $\pi_\Gamma: \R^d\rightarrow \Gamma$ the projection function from $\R^d$ to $\Gamma$, i.e., for any $x\in X$, 
\[
\pi_\Gamma(x) := \arg \min_{y\in \Gamma} d(x,y).
\]
If $\Gamma$ is clear from the context, we may denote $\pi_\Gamma$ by $\pi$.

\begin{definition}[\bf{Equivalence relations and equivalence classes induced by a subspace}]
	\label{def:equivalent}
	Given a subspace $\Gamma\subsetneq \R^d$, we define an equivalence relation $\sim_\Gamma$ between $k$-center sets as follows: for two $k$-center sets $C=\left(c_1,\ldots,c_k\right)$ and $C'=\left(c'_1,\ldots,c'_k\right)$, we say $C\sim_\Gamma C'$ if and only if for all $i\in [k]$,
	\[
	\pi_\Gamma(c_i) = \pi_\Gamma(c'_i) \quad \text{and} \quad d(c_i,\pi_\Gamma(c_i)) = d(c'_i,\pi_\Gamma(c'_i)).
	\]
	Let $\Gamma'$ be obtained from $\Gamma$ by appending an arbitrary dimension in $\R^d$ that is orthogonal to $\Gamma$.\footnote{Here, $\Gamma' = \left\{ ax+bu\mid x\in \Gamma, a\in \R,b\in \R\right\}$.}
	Let $\calC_\Gamma$ denote the collection of $k$-center sets $C$ all of whose points lie in $\Gamma'$, i.e.,
	\[
	\calC_\Gamma:= \left\{C=(c_1,\ldots,c_k) \in \mathcal{C}: c_i\in  \Gamma' \; \forall i\in[k]\right\}.
	\]
	The relation $\sim_\Gamma$ also induces equivalence classes of $k$-center sets $\left\{\Delta^{(\Gamma)}_C: C\in \calC_\Gamma \right\}$ where each $\Delta^{(\Gamma)}_C := \left\{C'\in \calC: C\sim_\Gamma C'\right\}$.
\end{definition}

\eat{
In the following, we explain why $\sim_\Gamma$ is an equivalence relation.
Given a subspace $\Gamma\subsetneq \R^d$, we define a mapping $\phi_\Gamma: \R^d \rightarrow \R^{d+1}$ where for any point $x\in \R^d$,
\[
\phi_\Gamma(x):= \left(\pi_\Gamma(x),d(x,\pi_\Gamma(x)\right).
\]
Given a $k$-center set $C=\left(c_1,\ldots,c_k\right)\in \calC$, we denote $\phi_\Gamma(C):= \left(\phi_\Gamma(c_1),\ldots,\phi_\Gamma(c_k)\right)$ to be a $k$-center set in $\R^{d+1}$.
Let $C=\left(c_1,\ldots,c_k\right)$ and $C'=\left(c'_1,\ldots,c'_k\right)$ be two center sets in $\calC$.
We have that $C\sim_\Gamma C'$ if and only if $\phi_\Gamma(C) = \phi_\Gamma(C')$, namely any pair of center points $c_i\in C$ and $c_i'\in C'$ satisfy that 1) the projections to $\Gamma$ are the same, i.e., $\pi_\Gamma(c_i) = \pi_\Gamma(c'_i)$; 2) the projection distances to $\Gamma$ are the same, i.e., $d(c_i,\pi_\Gamma(c_i)) = d(c'_i,\pi_\Gamma(c'_i))$.
This indicates that $\sim_\Gamma$ defines an equivalence relation.
}

\noindent
In the following, we explain why $\left\{\Delta^{(\Gamma)}_C: C\in \calC_\Gamma \right\}$ are equivalence classes induced by $\Gamma$.
Given a subspace $\Gamma\subsetneq \R^d$, we define a mapping $\phi: \R^d \rightarrow \R^{d+1}$ where for any point $x\in \R^d$,
\[
\phi(x):= \left(\pi_\Gamma(x),d(x,\pi_\Gamma(x)\right).
\]
Let $A: = \left\{\phi(C): C\in \calC\right\}$ denote the collection of images of $\phi$ with respect to $\calC$.
By the discussion above, each image $a\in A$ naturally corresponds to an equivalence class $\Delta_a^{(\Gamma)}:= \left\{C\in \calC: \phi(C)=a\right\}$.
Note that for any $x\in \R^d$, there must exist a point $x'\in \Gamma'$ such that
\[
\pi_\Gamma(x) = \pi_\Gamma(x') \quad \text{and} \quad d(x,\pi_\Gamma(x)) = d(x',\pi_\Gamma(x')),
\]
since $\Gamma'$ includes an additional dimension that is orthogonal to $\Gamma$. 
By the above equations, we have that $\phi(x) = \phi(x')$.
Thus, for any $a\in A$ and $C'\in \calC$ satisfying that $\phi(C')=a$, there must exist a $k$-center set $C\in \calC_\Gamma$ such that $\phi(C) = \phi(C')=a$.
This implies that $\Delta_a^{(\Gamma)}=\Delta^{(\Gamma)}_C$.
Since $\left\{\Delta^{(\Gamma)}_a: a\in A \right\}$ are equivalence classes induced by $\Gamma$, we conclude that $\left\{\Delta^{(\Gamma)}_C: C\in \calC_\Gamma \right\}$ are also equivalence classes induced by $\Gamma$.
\begin{example}
	\label{ex:equivalence}
	Let $d=3$, $k=1$ and $\Gamma$ denote the the first axis.
	Let $x=(x_1,x_2,x_3)$ and $x'=(x'_1,x'_2,x'_3)$ be two $1$-center sets in $\R^3$.
	By Definition~\ref{def:equivalent}, we have that $c\sim_\Gamma c'$ if and only if 
	\[
	x_1 = x'_1, \quad \text{and} \quad \sqrt{x_2^2 + x_3^2} = \sqrt{(x'_2)^2 + (x'_3)^2}.
	\]
	i.e., their first coordinates are the same and their distances to the first axis are the same.
	Hence, all points with the same $x$-coordinate form an equivalence class.

	Without loss of generality, let $\Gamma'$ be the plane spanned by the first and the second axes.
	Then each center $x = (x_1,x_2,0)\in \Gamma'$ corresponds to an equivalence class
	\[
	\Delta_x^{(\Gamma)} = \left\{x=(x_1,x_2,x_3)\in \R^d: x_1 = x_1, \sqrt{x_2^2 + x_3^2} = |x_2| \right\}.
	\]
\end{example}

\noindent
We also propose the following definition that indicates that \kzC objectives within an equivalence class are almost the same.

\begin{definition}[\bf{Representativeness property}]
	\label{def:representativeness}
	Given a weighted point set $A\subseteq \R^d$ together with a weight function $w: A\rightarrow \R_{\geq 0}$, $\eps\in (0,1)$ and a subspace $\Gamma$, we say that $A$ satisfies the $\eps$-representativeness property with respect to $\Gamma$ if for any equivalence class $\Delta_C^{\Gamma}$ and any two $k$-center sets $C_1,C_2\in \Delta_C^{\Gamma}$, the following property holds:
	\[
	\cost_z(A,C_1)\in (1\pm \eps)\cdot \cost_z(A,C_2).
	\]
\end{definition}

\noindent
It follows from the definition above that, if both the given dataset $X$ and a weighted point set $S$ satisfy the representativeness property with respect to a certain low-dimensional subspace $\Gamma$, then $S$ is a coreset if $S$ approximately preserves all \kzC objectives with respect to $k$-center sets $C\in \calC_\Gamma$. 
This observation enables us to only consider $k$-center sets in $\Gamma$ instead of $\R^d$, which is the key for our dimension reduction technique.

\paragraph{$(k,z)$-subspace approximation.} 
Let $\calP$ denote the collection of all $j$-flats in $\R^d$ with $j\leq k$, i.e., all subspaces in $\R^d$ of dimension at most $k$.

\begin{definition}[\bf{$(k,z)$-subspace approximation problem}]
	\label{def:subspace}
	Given a dataset $X$ in $\R^d$, $z>0$ and an integer $k\geq 1$, the goal of the $(k,z)$-subspace approximation problem is to find a subspace $P^\star\in \calP$ that minimizes $\sum_{x\in X} d^z(x,P)$ over all $P\in \calP$.
\end{definition}

\noindent
Subspace approximation is  a well-studied shape fitting problem. 
Several prior works have focussed on designing approximation algorithms~\cite{deshpande2011algorithms,chierichetti2017algorithms,ban2019ptas} and constructing coresets~\cite{feldman2011unified,cohen2015p,sohler2018strong} for subspace approximation.
In this paper, we propose the following version of ``weak-coreset'' for subspace approximation.

\begin{definition}[\bf{Weak-coreset for $(k,z)$-subspace approximation}]
	\label{def:weak_coreset}
	Given a collection $X\subseteq \R^d$ of $n$ weighted points and an $\eps\in (0,1)$, an $\eps$-weak-coreset for $(k,z)$-subspace approximation is a subset $S \subseteq \R^d$ with weights $w : S \rightarrow \R_{\geq 0}$ such that
	\begin{align}
	\label{eq:subspace}
	\min_{P\in \calP} \sum_{x\in S} w(x)\cdot d^z(x,P)\in (1\pm\eps)\cdot\min_{P\in \calP}\sum_{x\in X} d^z(x,P).
	\end{align}
\end{definition}

\noindent
Note that a weak-coreset may not approximately preserve all subspace approximation objectives like Definition~\ref{def:coreset}.
However, we can approximately compute the minimum $(k,z)$-subspace approximation objective via a weak-coreset.
We remark that Definition~\ref{def:weak_coreset} is a different version of a notion in~\cite{feldman2011unified}, in which a weak-coreset $S$ satisfies that any $(1+\eps)$-approximate solution for $(k,z)$-subspace approximation over $S$ is an $(1+O(\eps))$-approximate solution over $X$.

\section{Prior results on importance sampling and terminal embeddings}
\label{sec:result}

We first introduce two general frameworks for coresets for \kzC that will be used in our algorithm. 
Both were proposed by Feldman and Langberg~\cite{feldman2011unified}, and the second one is an improved version by~\cite{braverman2016new}; summarized as follows.

\begin{theorem}[\bf{\FL Framework~\cite{feldman2011unified,braverman2016new}}]
	\label{thm:fl11_bfl16}
	Let $\eps,\delta\in (0,1/2)$, $k\geq 1$ and constant $z\geq 1$.
	Let $X\subseteq \R^d$ denote a weighted point set of $n$ points together with a weight function $u: X\rightarrow \R_{\geq 0}$.
	Let $C^\star\in \calC$ denote a $k$-center set that is an $O(1)$-approximate solution for \kzC over $X$.
	We have two importance sampling frameworks as follows.
	\begin{enumerate}
	\item (\cite[Theorem 15.5]{feldman2011unified}) Suppose $\sigma:X\rightarrow \R_{\geq 0}$ is a sensitivity function satisfying that for any $x\in X$,
	\[
	\sigma(x) \geq \frac{ u(x)\cdot d^z(x,C^\star)}{\sum_{y\in X} u(y)\cdot d^z(y,C^\star)},
	\] 
	and $\calG := \sum_{x\in X}\sigma(x)$.
	Let $D\subseteq X$ be constructed by taking 
	\[
	O\left(\eps^{-2z} (dk \log k+\log(1/\delta))\right)
	\] 
	samples, where each sample $x\in X$ is selected with probability $\frac{\sigma(x)}{\calG}$ and has weight $w(x):= \frac{\calG}{|D|\cdot \sigma(x)}$.
	For each $c\in C^\star$, let $w(c):= (1+10\eps)\cdot\sum_{x\in X_c}u(x)-\sum_{x\in D\cap X_c} w(x)$ where $X_c$ is the collection of points in $X$ whose closest point in $C^\star$ is $c$.
	Then, with probability at least $1-\delta$, $S:=D\cup C^\star$ is an $\eps$-coreset for \kzC over $X$.\footnote{This conclusion is based on~\cite[Theorem 15.5]{feldman2011unified}. We discuss the theorem in Section~\ref{sec:fl11}.}
	\item (\cite[Theorem 5.2]{braverman2016new}) Suppose $\sigma:X\rightarrow \R_{\geq 0}$ is a sensitivity function satisfying that for any $x\in X$,
	\[
	\sigma(x) \geq \sup_{C\in \calC} \frac{u(x)\cdot d^z(x,C)}{\sum_{y\in X} u(y)\cdot d^z(y,C)},
	\] 
	and $\calG := \sum_{x\in X}\sigma(x)$.
	Let $S\subseteq X$ be constructed by taking 
	\[
	O\left(\eps^{-2} \calG (dk \log \calG +\log(1/\delta))\right)
	\] 
	samples, where each sample $x\in X$ is selected with probability $\frac{\sigma(x)}{\calG}$ and has weight $w(x):= \frac{\calG}{|S|\cdot \sigma(x)}$.
    Then, with probability at least $1-\delta$, $S$ is an $\eps$-coreset for \kzC over $X$.
	\end{enumerate}
\end{theorem} 

\noindent
The \FL framework applies one-staged importance sampling and is easy to implement.
We only need to compute an approximate solution $C^\star$, sensitivities $\sigma(x)$, and take samples with probability proportional to $\sigma(x)$.
For the first framework, note that we can further guarantee that the output $S$ is a subset of $X$ by adding an additional constraint that $C^\star\subset X$.\footnote{Suppose $C^\star$ is a $k$-center set that is an $O(1)$-approximate solution for \kzC over $X$. Then $\left\{\arg\min_{x\in X} c^\star_i : i\in [n] \right\}$ is a $k$-center set that is an $O(2^z)$-approximate solution for \kzC over $X$. Hence, we can add this additional constraint.}
For the second framework, the total sensitivity $\calG$ is shown to be $O(2^{2z} k)$ by~\cite{varadarajan2012sensitivity} and, hence, the resulting coreset size is quadratic in $k$ by plugging the value of $\calG = O(2^{2z} k)$ in the size bound $O\left(\eps^{-2} \calG (dk \log \calG +\log(1/\delta))\right)$ in Theorem~\ref{thm:fl11_bfl16}.

\paragraph{Terminal embedding.}
We introduce the definition of terminal embedding that is useful for our dimension reduction result.
Roughly speaking, a terminal embedding projects a point set $X\subseteq \R^d$ to a low-dimensional space such that all pairwise distances between $X$ and $\R^d$ are approximately preserved.

\begin{definition}[\bf{Terminal embedding}]
	\label{def:embedding}
	Let $\eps\in (0,1)$ and $X\subseteq \R^d$ be a collection of $n$ points. 
	A mapping $f: \R^d\rightarrow \R^m$ is called a terminal embedding of $X$ if for any $x\in X$ and $y\in \R^d$,
	\[
	d(x,y)\leq d(f(x),f(y))\leq (1+\eps)d(x,y).
	\]
\end{definition}

\noindent
Note that a terminal embedding must be a one-to-one mapping over $X$ by definition.
The following is a recent result on terminal embedding.

\begin{theorem}[\bf{Small terminal embeddings~\cite{narayanan2019optimal}}]
	\label{thm:embedding}
	Let $\eps\in (0,1)$ and $X\subseteq \R^d$ be a collection of $n$ points. 
	There exists a terminal embedding with a target dimension $m=O(\eps^{-2}\log n)$.
\end{theorem}

\noindent
It follows from this theorem that, if $|X|=\poly(k/\eps)$, then there is a terminal embedding of target dimension $O\left(\eps^{-2} \log(k/\eps)\right)$,  which is independent of $d$.
This  is useful for analyzing the correctness of the second stage of our framework.
%

\section{Technical overview}
\label{sec:overview}

%
In this section, we will show how to prove our algorithmic result (Theorem~\ref{thm:coreset_informal}) and how to achieve a nearly matching size lower bound (Theorem~\ref{thm:lower}).
We first introduce a new dimension reduction technique that is useful for our algorithmic result.
%

\paragraph{Dimension reduction.}
For simplicity, we take the \kMedian problem ($z=1$) as an example and let $X\subseteq \R^d$ be the input point set.
By Theorem~\ref{thm:fl11_bfl16}, there exists an $\eps$-coreset for \kMedian of size $O(\eps^{-2} dk \log k)$.
However, the coreset size contains a factor $d$ and our goal is to construct a coreset that does not depend on $d$.
%

To this end, a commonly used approach is called \emph{dimension reduction}~\cite{feldman2013turning,cohen2015dimensionality,braverman2016new,sohler2018strong}.
Roughly speaking, we would like to show that it suffices to only consider all $k$-center sets in some low-dimensional space instead of $\R^d$, and that enables us to remove the dependence on $d$ in the coreset size.
One can try the dimension reduction approach proposed by~\cite{makarychev2019performance}, however, this can be shown not to work.
We explain the details in Section~\ref{sec:failed}.
%


\paragraph{Representativeness property can allow to remove the dependence on $d$.} Our second attempt is motivated by the work of~\cite{sohler2018strong} and projects $X$ to a low-dimensional subspace such that all \kMedian objectives can be estimated by the projections of $X$.
For a subspace $\Gamma\subsetneq \R^d$, recall that $\Gamma'$ is obtained from $\Gamma$ by appending an arbitrary dimension in $\R^d$ that is orthogonal to $\Gamma$.
Also recall that $\calC_\Gamma$ denotes the collection of $k$-center sets $C\subset \Gamma'$, i.e.,
\[
\calC_\Gamma:= \left\{C=(c_1,\ldots,c_k): c_i\in  \Gamma' \; \forall i, C\in \calC\right\}
\]
Now suppose we have a subspace $\Gamma\subsetneq \R^d$ of dimension $\poly(k/\eps)$ and a weighted point set $D$, together with a weight function $u: D\rightarrow \R_{\geq 0}$, that approximately preserves all \kMedian objectives with respect to $k$-center sets in $\calC_\Gamma$, i.e., for any $k$-center set $C\in \calC_\Gamma$,
\begin{align}
\label{eq:SW2}
\sum_{x\in D} u(x)\cdot d(x,C) \in (1\pm \eps)\cdot \cost_1(X,C).
\end{align}
By Theorem~\ref{thm:fl11_bfl16}, the size of $D$ can be upper bounded by $\poly(k/\eps)$.
If we can show that Inequality~\eqref{eq:SW2} holds for all $k$-center sets in $\calC$, then we obtain an $\eps$-coreset $D$ of size independent of $d$.
Moreover, by the result on terminal embeddings (Theorem~\ref{thm:reduction}), we can further reduce the dimension to $O\left(\eps^{-2} \log(k/\eps)\right)$.
This observation leads to an $\eps$-coreset $S$ for \kMedian of size $\tilde{O}(\eps^{-4} k)$ by applying the first framework in Theorem~\ref{thm:fl11_bfl16} to $D$.
In the following, we show how to construct sufficient conditions for $\Gamma$ and $D$ such that Inequality~\eqref{eq:SW2} holds for all $k$-center sets in $\calC$, which leads to Theorem~\ref{thm:upper}.

Our key idea is to construct a $\Gamma$ such that both $X$ and $D$ satisfy the representativeness property with respect to $\Gamma$.
Recall that an equivalence class $\Delta^\Gamma_{C'}$ for a center set $C'\in \calC_\Gamma$ is the collection of center sets $C\in \calC$ such that $C\sim_\Gamma C'$. 
Now suppose we have that Inequality~\eqref{eq:SW2} holds for all $k$-center sets in $\calC_\Gamma$.
Consequently, for any $C\in \Delta^\Gamma_{C'}$,
\[
\sum_{x\in D} u(x)\cdot d(x,C) \approx \sum_{x\in D} u(x)\cdot d(x,C') \approx \cost_1(X,C') \approx \cost_1(X,C).
\]
Then $D$ is an $\eps$-coreset for \kMedian in $\R^d$.
Hence, we focus on proving the representativeness property.
Given a subset $A\subseteq \R^d$, let $\mathrm{Span}(A)$ denote the subspace spanned by $A$.
Recall that $\OPT_1$ denotes the optimal \kMedian objective over $X$.
Sohler and Woodruff~\cite[Algorithm 1]{sohler2018strong} show how to construct a subspace $\Gamma$ of dimension $O(\eps^{-2} k)$ such that for any $k$-center set $C\in \calC$, letting $\pi_C$ denote the projection from $X$ to $\mathrm{Span}(\Gamma\cup C)$ (the subspace spanned by $\Gamma \cup C$),
\begin{align}
\label{eq:SW1}
\sum_{x\in X} d(x,\pi_\Gamma(x))-d(x,\pi_C(x))= O(\eps^2) \cdot \OPT_1.
\end{align}
By~\cite{sohler2018strong} (summarized in Lemma~\ref{lm:projection}), this implies that for any $k$-center set $C\in \calC$,
\begin{align}
\label{eq:offset}
\sum_{x\in X} \left(d^2(\pi_\Gamma(x),C)+d^2(x,\pi_\Gamma(x))\right)^{1/2} \in (1\pm \eps)\cdot \cost_1(x,C),
\end{align}
i.e., we can use $\Gamma$ to approximately preserve all \kMedian objectives.
Note that for any two center sets $C_1,C_2$ in an equivalence class, we have
\[
\sum_{x\in X} \left(d^2(\pi_\Gamma(x),C_1)+d^2(x,\pi_\Gamma(x))\right)^{1/2} = \sum_{x\in X} \left(d^2(\pi_\Gamma(x),C_2)+d^2(x,\pi_\Gamma(x))\right)^{1/2}.
\]
Combining the above equation with Inequality~\eqref{eq:offset}, we conclude that $X$ satisfies the representativeness property with respect to $\Gamma$ (Observation~\ref{ob:equivalent}).
The technical difficulty is to show that $D$ also satisfies the representativeness property with respect to $\Gamma$.
%

\paragraph{A failed attempt.}
By a similar construction as~\cite[Algorithm 1]{sohler2018strong}, we can construct a subspace $\Gamma$ for any given $D$ that satisfies for any $k$-center set $C\in \calC$,
\begin{align}
\label{eq:SW3}
\sum_{x\in D} u(x)\cdot \left(d^2(\pi_\Gamma(x),C)+d^2(x,\pi_\Gamma(x))\right)^{1/2} \in (1\pm \eps)\cdot \cost_1(D,C),
\end{align}
i.e., we can also use $\Gamma$ to approximately preserve \kMedian objectives over the weighted point set $D$.
Similarly, we conclude that $D$ also satisfies the representativeness property with respect to $\Gamma$.
Then assuming Inequality~\eqref{eq:SW2} holds for all $k$-center sets $C\subset \Gamma'$, $D$ is an $O(\eps)$-coreset for \kMedian and we are done.
However, this assumption only has a guarantee when $\Gamma$ is independent of the choice of $D$.
Thus, our task is to construct a subspace $\Gamma$ satisfying Inequalities~\eqref{eq:SW1} and~\eqref{eq:SW3}, and meanwhile independent on the choice of $D$.
%

\paragraph{Weak coreset for subspace approximation implies the representativeness property for $D$.}
By~\cite[Algorithm 1]{sohler2018strong}, we construct a subspace $\Gamma$ of dimension $\poly(k/\eps)$ that satisfies Inequality~\eqref{eq:SW1} and only depends on $X$, which implies that $\Gamma$ is independent on the choice of $D$.
The remaining task is to find a sufficient condition for $D$ such that Inequality~\eqref{eq:SW3} holds.
Again by~\cite{sohler2018strong} (summarized in Lemma~\ref{lm:projection}), Inequality~\eqref{eq:SW3} holds if for any $k$-center set $C\in \calC$,
\begin{align}
\label{eq:SW4}
\sum_{x\in D} u(x)\cdot \left(d(x,\pi_\Gamma(x))-d(x,\pi_C(x))\right)= O(\eps^2) \cdot \OPT_1.
\end{align}
Since Inequality~\eqref{eq:SW1} holds, for the above inequality it suffices to guarantee that
\begin{align}
\label{eq:approx}
\sum_{x\in D} u(x)\cdot \left(d(x,\pi_\Gamma(x))-d(x,\pi_C(x))\right) \approx \sum_{x\in X} d(x,\pi_\Gamma(x))-d(x,\pi_C(x)).
\end{align}
Our first assumption is that the total projection distance to $\Gamma$ is approximately preserved by $D$:
\[
\sum_{x\in D} u(x)\cdot d(x,\pi_\Gamma(x)) \approx \sum_{x\in X} d(x,\pi_\Gamma(x)).
\] 
Then for Inequality~\eqref{eq:approx}, it suffices to ensure that 
\begin{align}
\label{eq:SW5}
\min_{C\in \calC}\sum_{x\in D} u(x)\cdot d(x,\pi_C(x))\approx \min_{C\in \calC}\sum_{x\in X} d(x,\pi_C(x)).
\end{align}
Interestingly, this relates to the subspace approximation problem by regarding $X$ as a point set in space $\Gamma^\perp$.
From this viewpoint, Inequality~\eqref{eq:SW5} can be reduced to ensuring that $D$ is a weak-coreset for $(k,1)$-subspace approximation over $X$ (Definition~\ref{def:weak_coreset}).
Overall, roughly, a weak-coreset $D$ for $(k,1)$-subspace approximation satisfies the representativeness property with respect to $\Gamma$.

We still need to verify that the size of a weak-coreset $D$ can be independent of  dimension $d$. 
By applying~\cite{shyamalkumar2007efficient}, we know that $D$ is a weak-coreset if $D$ approximately preserves all $(k,1)$-subspace approximation objectives in the collection $\calP'$ of all $k$-flats spanned by at most $\tilde{O}(\eps^{-1} k^2)$ points from $X$ (Lemma~\ref{lm:projective_clustering}), i.e., $D$ is an $\eps$-coreset for $(k,1)$-subspace approximation over $X$ with respect to $\calP'$.
By~\cite[Theorem 4]{varadarajan2012sensitivity}, it follows that the coreset size $|D|$ only depends on $k,\eps$ and the ``function dimension'' of $\calP'$\footnote{Since this paper only uses function dimension as a black box, we do not present the definition. We refer interested readers to~\cite[Definition 6.4]{feldman2011unified} or~\cite[Definition 4.5]{braverman2016new} for concrete definitions.}
Since the function dimension of $\calP'$ is $\poly(k/\eps)$ by~\cite[Theorem 4]{varadarajan2012sensitivity}, the size of a weak-coreset $|D|$ can be also upper bounded by $\poly(k/\eps)$, which is independent of $d$ (Theorem~\ref{thm:subspace}).

Overall, we  develop a two-stage importance sampling framework that constructs an $\eps$-coreset for \kMedian of size $\tilde{O}(\eps^{-4} k)$ (Algorithm~\ref{alg:coreset}).
In the first stage, we construct a weighted point set $D$ of size $\poly(k/\eps)$ that is a coreset for \kMedian in $\Gamma'$ and an $\eps$-weak-coreset for $(k,1)$-subspace approximation in $\Gamma^\perp$.
By the discussion above, $D$ is an $\eps$-coreset for \kMedian in $\R^d$ of size $\poly(k/\eps)$ (Theorem~\ref{thm:upper}).
In the second stage, we further construct an $\eps$-coreset $S$ over $D$ of size $\tilde{O}(\eps^{-4} k)$ using the result on terminal embeddings (Theorem~\ref{thm:reduction}), where $S$ is also an $O(\eps)$-coreset for \kMedian over $X$ (Theorem~\ref{thm:coreset_informal}).
For general $z> 1$, the proof is similar to $z=1$.
The  difference is that in the first stage of Algorithm~\ref{alg:coreset}, we need the weighted point set $D$ to be a coreset for \kzC in $\Gamma'$ and an $\eps$-weak-coreset for $(k,z)$-subspace approximation in $\Gamma^\perp$.
To achieve this, the number of samples $|D|$ should be $z^{O(z)} \cdot \poly(k/\eps)$, which is still independent of $d$.
Again, using the result on terminal embeddings, we can construct a coreset for \kzC of size $\tilde{O}(\eps^{-2z-2} k)$ (Theorem~\ref{thm:coreset_informal}).
%

\paragraph{Size lower bounds.}
We also provide a nearly matching size lower bound of coresets for \kzC (Theorem~\ref{thm:lower}), by constructing a point set $X \subset \mathbb{R}^d$ in $\R^d$ such that any $0.01$-coreset for \kzC over $X$ has size $\Omega\left(k\cdot \min\left\{2^{z/20},d\right\}\right)$.
The main idea is to ensure that for any point $x\in X$, there exists a $k$-center set $C_x$ satisfying that the clustering objective 
$
d^z(x,C_x)\approx \cost_z(X,C_x).
$
Intuitively, we need to include a close point for each $x\in X$ in any coreset such that $\cost_z(X,C_x)$ can be approximately preserved.

We discuss a simple case that $k=1$ and $d\approx 2^{z/20}$, and show how to construct such a bad instance $X$.
The general case that $k\geq 1$ can be proved by making $k$ copies of $X$ in which all copies are far away from each other.
We let $X=\left\{e_1,-e_1,\ldots, e_d,-e_d\right\}$ and observe that for each $e_i\in X$, 
\[
d^z(e_i,-e_i)\approx \cost_z(X,-e_i).
\]
Suppose $S$ is a $0.01$-coreset for \kzC together with weights $w(x)$.
If we restrict that $S\subseteq X$, we can conclude that $S=X$ which implies that $|S|= 2^{d+1}=\Omega(2^{z/20})$.
Intuitively, if there exists $e_i\in X\setminus S$, then it is unlikely that the \kzC objective with respect to center $-e_i$ can be approximately preserved by $S$.
The obstacle is that points in a coreset can come from outside of $X$.

We divide $\R^d$ into $|X|$ Voronoi cells $(P_x)_{x}$ induced by $X$, where each $P_x$ is the collection of points in $\R^d$ whose closest point in $X$ is $x$.
If $|S|< |X|$, there must exist a Voronoi cell that does not contain points in $S$, say $P_{e_1}$ without loss of generality.
The key idea is to show that $S$ can not approximately preserve the \kzC objectives with respect to center $0$ and center $-e_1$ simultaneously, which implies that the size of $S$ should be at least $|X|=2d$.
Let $H$ denote the unit $l_2$-ball centered at the origin. 
On the one hand, we can show that the contribution $\sum_{x\in S\cap H} w(x)\cdot d^z(x,-e_1)$ is tiny compared to $\sum_{x\in X}d^z(x,-e_1)\approx 2^z$, since each point $x\in H\setminus P_{e_1}$ satisfies that 
\[
d^z(x,-e_1)\leq 2^{0.9z} \qquad (\text{Inequality~\eqref{eq:lower3}})
\] 
and the total weights $\sum_{x\in S\cap H} w(x)$ can be shown at most $2^{0.05 z}$ (Claim~\ref{claim:lower1}).
On the other hand, for point $x\in S\setminus (H\cup P_{e_1})$, we can verify that
\[
\frac{d^z(x,-e_1)}{d^z(x,0)}\leq 2^{0.9z}. \qquad \text{(Inequality~\eqref{eq:lower4})}
\]
This gap is much smaller compared to $\frac{\sum_{x\in X} d^z(x,-e_1)}{\sum_{x\in X} d^z(x,0)} \approx 2^{0.95 z}$, which implies that $S$ can not approximately preserve the \kzC objectives with respect to center 0 and center $-e_1$ simultaneously.
This indicates that any $0.01$-coreset for \kMedian over $X$ has size at least $|X|=2d$, which proves Theorem~\ref{thm:lower}.

\begin{algorithm}[ht]
	\caption{Coreset construction for \kzC}
	\label{alg:coreset}
	\KwIn{A dataset $X$ of $n$ points in $\R^d$, $\eps, \delta \in (0,1/2)$, constant $z\geq 1$ and an integer $k \geq 1$.}
	\KwOut{A point set $S\subseteq \R^d$ together with a weight function $w:S\rightarrow \R_{\geq 0}$.} 
	\nonl \hrulefill \\
	\tcc{The first importance sampling stage}
	$N_1 \leftarrow O\left((168z)^{10z} \eps^{-5z-15}  k^5 \log \frac{k}{\delta}  \right)$ \;
	Compute a $k$-center set $C^\star\subseteq \R^d$ as an $\alpha$-approximation of the \kzC problem over $X$ ($\alpha=O(1)$)\;
	For each $x\in X$, compute $c^\star(x)$ to be the closest point to $x$ in $C^\star$ (ties are broken arbitrarily). For each $c\in C^\star$, denote $X_c$ to be the set of points $x\in X$ with $c^\star(x)=c$\;
	For each $x\in X$, let $\sigma_1(x)\leftarrow 2^{2z+2} \alpha^2 \left(\frac{d^z(x,c^\star(x))}{ \cost_z(X,C^\star)}+\frac{1}{|X_{c^\star(x)}|}\right)$\;
	Pick a non-uniform random sample $D_1$ of $N_1$ points from $X$, where each $x\in X$ is selected with probability $\frac{\sigma_1(x)}{\sum_{y\in X}\sigma_1(y)}$.
	For each $x\in D_1$, let $u(x)\leftarrow \frac{\sum_{y\in X}\sigma_1(y)}{|D_1|\cdot \sigma_1(x)}$\;
	%
	%
	\nonl \hrulefill \\
	\tcc{The second importance sampling stage}
	$N_2\leftarrow O\left(\eps^{-2z-2} k\log k \log \frac{k}{\eps \delta}\right)$\;
	For each $c\in C^\star$, compute $D_c$ to be the set of points in $D_1$ whose closest point in $C^\star$ is $c$ (ties are broken arbitrarily)\;
	For each $x\in D_1$, let $\sigma_2(x)\leftarrow \frac{u(x)\cdot d^z(x,C^\star)}{\sum_{y\in D_1} u(y)\cdot d^z(y,C^\star)}$\;
	Pick a non-uniform random sample $D_2$ of $N_2$ points from $X$, where each $x\in X$ is selected with probability $\frac{\sigma_2(x)}{\sum_{y\in D_1}\sigma_2(y)}$.
	For each $x\in D_2$, let $w(x)\leftarrow \frac{\sum_{y\in D_1}\sigma_2(y)}{|D_2|\cdot \sigma_2(x)}$.
	For each $c\in C^\star$, let $w(c)\leftarrow (1+10\eps)\sum_{x\in D_c}u(x)-\sum_{x\in D_2\cap D_c} w(x)$\;
	$S\leftarrow D_2\cup C^\star$\;
	Output $(S,w)$\;
\end{algorithm}

\section{Our algorithm and main theorem}
\label{sec:upper}

Let $X\subseteq \R^d$ denote a collection of $n$ points in $\R^d$.
Let $\OPT_z$ denote the optimal \kzC objective over $X$.
The main theorem of this paper is the following.

\begin{theorem}[\bf{Coreset for \kzC; near-linear size in $k$}]
	\label{thm:coreset}
	There exists a randomized algorithm that, for a given dataset $X$ of $n$ points in $\R^d$, $\eps,\delta \in (0,0.5)$, constant $z\geq 1$ and integer $k \geq 1$, with probability at least $1-\delta$, constructs an $\eps$-coreset for \kzC of size 
	\[
	O\left(\min\left\{\eps^{-2z-2}, 2^{2z} \eps^{-4} k \right\} k\log k \log\frac{k}{\eps\delta}\right)
	\] 
	and runs in time 
	\[
	O\left(ndk+nd \log(n/\delta) + k^2 \log^2 n + \log^2(1/\delta) \log^2 n\right).
	\]
\end{theorem}

\noindent
We propose a unified two-staged importance sampling framework for \kzC; see Algorithm~\ref{alg:coreset}.
Note that Algorithm~\ref{alg:coreset} provides an $\eps$-coreset of size $\tilde{O}(\eps^{-2z-2} k)$.
By applying another importance sampling approach in the second stage (see discussion in Remark~\ref{remark:another_alg}), we can achieve an $\eps$-coreset of size $\tilde{O}(2^{2z}\eps^{-4} k^2)$.
This gives the size bound stated in Theorem~\ref{thm:coreset}.

\subsection{Proof of the main Theorem~\ref{thm:coreset}}
\label{sec:analysis}

The proof of Theorem \ref{thm:coreset} relies on  the following two theorems.
The first theorem shows that the first stage of Algorithm~\ref{alg:coreset} constructs a coreset for \kzC of size $N_1=\poly(k,1/\eps)$ with high probability.
The second one is a size reduction theorem that constructs a coreset of size $N_2=\tilde{O}(\eps^{-2z-2} k)$ based on the output of the first stage.

\begin{theorem}[\bf{First stage of Algorithm~\ref{alg:coreset}}]
	\label{thm:upper}
	For every dataset $X$ of $n$ points in $\R^d$, $\eps,\delta \in (0,0.5)$, constant $z\geq 1$ and integer $k \geq 1$, with probability at least $1-\delta/2$, the first stage of Algorithm~\ref{alg:coreset} constructs an $\eps$-coreset $D_1$ for \kzC, and runs in time \[
	O\left(ndk+nd \log(n/\delta) + k^2 \log^2 n + \log^2(1/\delta) \log^2 n\right).
	\]
\end{theorem}

\begin{theorem}[\bf{Second stage of Algorithm~\ref{alg:coreset}}]
	\label{thm:reduction}
	Let $X\subseteq \R^d$ be a collection of $n$ points, $\eps,\delta \in (0,1/2)$, constant $z\geq 1$ and integer $k\geq 1$. 
	Suppose the first stage of Algorithm~\ref{alg:coreset} constructs an $\eps$-coreset $D_1$ for \kzC.
	Then with probability at least $1-\delta/2$, the second stage of Algorithm~\ref{alg:coreset} outputs an $O(\eps)$-coreset $S$ for \kzC, and runs in time $O(nd)$.
\end{theorem}

\noindent
The proof of Theorem~\ref{thm:upper} can be found in Section~\ref{sec:first_stage} and the proof of Theorem~\ref{thm:reduction} can be found in Section~\ref{sec:second_stage}.
Observe that Theorem~\ref{thm:coreset} is a direct corollary of Theorems~\ref{thm:upper} and~\ref{thm:reduction}. 

\begin{proof}[Proof of Theorem~\ref{thm:coreset}]
	The overall running time is a direct corollary of Theorems~\ref{thm:upper} and~\ref{thm:reduction}.

	By Theorem~\ref{thm:upper}, with probability at least $1-\delta/2$, $(D_1,u)$ is an $\eps$-coreset for \kzC of size $N_1$.
	Then by Theorem~\ref{thm:reduction}, with probability at least $1-\delta$, the output $S$ is an $\eps$-coreset which completes the proof.
\end{proof}

\subsection{Analyzing the first stage of Algorithm~\ref{alg:coreset}}
\label{sec:first_stage}

In this section, we prove Theorem~\ref{thm:upper} that provides a theoretical guarantee for the first stage.
Given a subset $A\subseteq \R^d$, let $\mathrm{Span}(A)$ denote the convex hull of $A$.
For preparation, we have the following lemmas.
The first lemma shows that there exists a subspace $\Gamma$ such that the projections of $X$ to $\Gamma$ can be used to estimate all \kzC objectives.
Note that~\cite{sohler2018strong} only considers unweighted point sets, but it can be easily generalized to weighted point sets.

\begin{lemma}[\textbf{Restatement of~\cite[Lemma 6 and Theorem 10]{sohler2018strong}}]
	\label{lm:projection}
	Let $A\subseteq \R^d$ be a weighted point set together with a weight function $w:A\rightarrow \R_{\geq 0}$.
	Let $\OPT_z$ be the optimal weighted $(k,z)$-clustering objective over $A$, $\eps\in (0,0.5)$, constant $z\geq 1$, and $\Gamma$ be a subspace of $\R^d$.
	Suppose for any $k$-center set $C\in \calC$, we have
	\begin{eqnarray}
	\label{eq:projection1}
	\sum_{x\in A} w(x)\cdot \left(d^z(x,\pi(x))-d^z(x,\pi_C(x))\right)\leq \frac{ \eps^{z+3}}{3\cdot (84z)^{2z}} \cdot \OPT_z,
	\end{eqnarray}
	where $\pi$ and $\pi_C$ denote the projection from $A$ to $\Gamma$ and $\mathrm{Span}(\Gamma\cup C)$ respectively.
	Then for any $k$-center set $C\in \calC$, the following inequality holds
	\begin{eqnarray}
	\label{eq:projection2}
	\sum_{x\in A} w(x)\cdot \left(d^2(\pi(x),C)+d^2(x,\pi(x))\right)^{z/2} \in (1\pm \eps)\cdot \sum_{x\in A} w(x)\cdot d^z(x,C).
	\end{eqnarray}
\end{lemma}

\noindent
we have the following lemma showing that $D_1$ is a coreset in an arbitrary low dimensional subspace with high probability.

\begin{lemma}[\bf{$D_1$ is a coreset in a low-dimensional subspace}]
	\label{lm:vx12}
	Suppose $\Gamma$ is an arbitrary subspace of dimension $O\left((84z)^{2z}\eps^{-z-3} k\right)$ in $\R^d$. 
	Let $\Gamma'$ be obtained from $\Gamma$ by appending an arbitrary dimension in $\R^d$ that is orthogonal to $\Gamma$.
	With probability at least $1-\delta/4$, for any $k$-center set $C\subset \Gamma'$ we have
	\[
	\sum_{x\in D_1} u(x)\cdot d^z(x,C) \in (1\pm \eps) \cdot \cost_z(X,C).
	\]
\end{lemma}

\begin{proof}
	Let $\eps' = \frac{ \eps^{z+3}}{6\cdot (84z)^{2z}}$ and $m=O(k/\eps')$ be the dimension of $\Gamma$.
	We first have the following claim.
	\begin{claim}
		\label{claim:sen_bound}
		For any $x\in X$,
		\[
		\sup_{C\in \calC} \frac{d^z(x,C)}{\cost_z(X,C)}\leq 2^z\cdot \frac{d^z(x,c^\star(x))}{\OPT_z}+ 2^{2z+1}\alpha\cdot \frac{1}{|X_{c^\star(x)}|}.
		\]
	\end{claim}
    \begin{proof}
    The proof idea is similar to~\cite[Theorem 7]{varadarajan2012sensitivity}.
	We first note that $d^z$ satisfies the relaxed triangle inequality, i.e., for any $x,x',x''\in \R^d$, we have
	\begin{align}
	\label{eq:triangle}
	d^z(x,x'')\leq 2^z\cdot\left(d^z(x,x')+d^z(x',x'')\right).
	\end{align}
	Then for any $x\in X$ and any $k$-center set $C\in \calC$, we have
	\begin{eqnarray}
	\label{eq:sen1}
	\begin{split}
	d^z(x,C)  &\leq && 2^z\left(d^z(x,c^\star(x)) + d^z(c^\star(x),C) \right) && (\text{Ineq.~\eqref{eq:triangle}}) \\
	& \leq && 2^z\cdot d^z(x,c^\star(x)) + \frac{2^z}{|X_{c^\star(x)}|}\cdot \sum_{y\in X} d^z(c^\star(y),C) && (\text{Defn. of $X_{c^\star(x)}$}) \\
	& \leq && 2^z\cdot d^z(x,c^\star(x)) + \frac{2^z}{|X_{c^\star(x)}|}\cdot \sum_{y\in X} 2^z\cdot \left(d^z(c^\star(x),x)+ d^z(x,C)\right) && (\text{Ineq.~\eqref{eq:triangle}}) \\
	& = && 2^z\cdot d^z(x,c^\star(x)) + \frac{2^{2z}}{|X_{c^\star(x)}|}\cdot \left(\cost_z(X,C^\star) + \cost_z(X,C) \right). && (\text{Defn. of $\cost_z$})
	\end{split}
	\end{eqnarray}
	Thus, we have that
	\begin{eqnarray*}
	\begin{split}
	\frac{d^z(x,C)}{\cost_z(X,C)} & \leq && 2^z\cdot \frac{d^z(x,c^\star(x))}{\cost_z(X,C)} + \frac{2^{2z}}{|X_{c^\star(x)}|}\cdot (1+\frac{\cost_z(X,C^\star)}{\cost_z(X,C)}) && (\text{Ineq.~\eqref{eq:sen1}}) \\
	&\leq && 2^z\cdot \frac{d^z(x,c^\star(x))}{\OPT_z} + \frac{2^{2z}}{|X_{c^\star(x)}|}\cdot (1+\alpha)  && (\text{Defn. of $C^\star$}) \\
	&\leq && 2^z\cdot \frac{d^z(x,c^\star(x))}{\OPT_z}+ 2^{2z+1}\alpha\cdot \frac{1}{|X_{c^\star(x)}|},
	\end{split}
	\end{eqnarray*}
	which implies the claim since $C$ is arbitrary.
	\end{proof}
	Then for any $x\in X$,
	\begin{eqnarray*}
	\begin{split}
	\sup_{C\in \calC} \frac{d^z(x,C)}{\cost_z(X,C)}
	& \leq && 2^z\cdot \frac{d^z(x,c^\star(x))}{\OPT_z}+ 2^{2z+1}\alpha\cdot \frac{1}{|X_{c^\star(x)}|} && (\text{Claim~\ref{claim:sen_bound}})\\
	& \leq && 2^z\alpha\cdot \frac{d^z(x,c^\star(x))}{ \cost_z(X,C^\star)}+ 2^{2z+1} \alpha\cdot \frac{1}{|X_{c^\star(x)}|} && (\text{Defn. of $C^\star$})\\
	&\leq && \sigma_1(x). &&
	\end{split}
	\end{eqnarray*}
	Also note that
	\begin{eqnarray*}
	\begin{split}
	\sum_{x\in X}\sigma_1(x) &= &&2^{2z+2} \alpha^2\cdot \sum_{x\in X} \left(\frac{d^z(x,c^\star(x))}{ \cost_z(X,C^\star)}+\frac{1}{|X_{c^\star(x)}|}\right) &&\\
	&\leq && 2^{2z+2} \alpha^2\cdot(1+k) && (|C^\star|=k) \\
	&\leq && 2^{2z+3} \alpha^2k. &&
	\end{split}
	\end{eqnarray*}
	Thus, we have
	\[
	N_1=\Omega\left(\frac{\sum_{x\in X}\sigma_1(x)}{(\eps')^2}\cdot(km \log (\sum_{x\in X}\sigma_1(x))+\log \frac{1}{\delta})\right),
	\] 
	Then by Theorem~\ref{thm:fl11_bfl16}, we complete the proof.
\end{proof}

\noindent
Next, we give the main technical lemma.
It indicates that if a subspace $\Gamma$ satisfies Inequality~\eqref{eq:projection1}, then clustering objectives over $D_1$ can be estimated by the projections of $D_1$ to $\Gamma$, similar to Inequality~\eqref{eq:projection2}.
The proof can be found in Section~\ref{sec:proof_mainlm}.

\begin{lemma}[\bf{$\Gamma$ preserves \kzC objectives over $D_1$}]
	\label{lm:preserve}
	Suppose $\Gamma$ is a subspace of $\R^d$ satisfying that $C^\star\subset \Gamma$ and for any $k$-center set $C\in \calC$, 
	\[
	\sum_{x\in X} \left(d^z(x,\pi(x))-d^z(x,\pi_C(x))\right)= \frac{ \eps^{z+3}}{3\cdot (84z)^{2z}} \cdot \OPT_z,
	\]
	where $\pi$ and $\pi_C$ denote the projection from $X$ to $\Gamma$ and $\mathrm{Span}(\Gamma\cup C)$ respectively.
	Let $D_1$ together with $u$ be the weighted point set obtained by the first stage of Algorithm~\ref{alg:coreset}.
	With probability at least $1-\delta/4$, for any $k$-center set $C\in \calC$,
	\[
	\sum_{x\in D_1} u(x)\cdot d^z(x,C)\in (1\pm 2\eps)\cdot \sum_{x\in D_1} u(x)\cdot \left(d^2(\pi(x),C)+d^2(x,\pi(x))\right)^{z/2}.
	\]
\end{lemma}

\noindent
By the above lemmas, we are ready to prove Theorem~\ref{thm:upper}.
The proof idea is to first show the existence of a subspace $\Gamma$ that satisfies Inequality~\eqref{eq:projection1} for $A=X$ in Lemma~\ref{lm:projection}. 
By Lemma~\ref{lm:projection}, we can prove that $X$ satisfies the representativeness property with respect to $\Gamma$.
Similarly, we can also show that $D_1$ satisfies the representativeness property with respect to $\Gamma$ by Lemma~\ref{lm:preserve}.
Recall that $\Gamma'$ is obtained from $\Gamma$ by appending an arbitrary  dimension in $\R^d$ that is orthogonal to $\Gamma$.
Finally, by Lemma~\ref{lm:vx12}, $D_1$ is an $\eps$-coreset for \kzC in $\Gamma'$.
Combining with these properties, we can conclude that $D_1$ is an $\eps$-coreset for \kzC in $\R^d$, which proves the theorem.

\begin{proof}[Proof of Theorem~\ref{thm:upper}]
	Let $\eps' = \frac{ \eps^{z+3}}{6\cdot (84z)^{2z}}$.
	We slightly modify~\cite[Algorithm 1]{sohler2018strong} by initiating $\Gamma$ to be the subspace containing $C^\star$ (whose dimension is at most $k$) instead of $\emptyset$. 
	Then by~\cite[Algorithm 1]{sohler2018strong}, there exists a subspace $\Gamma$ satisfying the following properties: 
	\begin{enumerate}
		\item[P1.] $\Gamma$ satisfies that for any $C\in \calC$,
		\[
		\sum_{x\in X} \left(d^z(x,\pi(x))-d^z(x,\pi_C(x))\right)\leq \eps' \cdot \OPT_z/2,
		\] 
		where $\pi$ and $\pi_C$ denote the projection from $X$ to $\Gamma$ and $\mathrm{Span}(\Gamma\cup C)$ respectively.
		\item[P2.] $\Gamma$ is of dimension $O(k/\eps')$.
		\item[P3.] $C^\star\subset \Gamma$.
	\end{enumerate}
	Since the dimension of $\Gamma'$ is $O(k/\eps')$, we have that with probability at least $1-\delta/4$, $D_1$ is an $\eps$-coreset for \kzC in $\Gamma'$ by Lemma~\ref{lm:vx12}.
	It means that for any $k$-center set $C\subset \Gamma'$,
	\begin{align}
	\label{eq:Q2}
	\sum_{x\in D_1} u(x)\cdot d^z(x,C) \in (1\pm \eps)\cdot \cost_z(X,C).
	\end{align}
	Moreover, with probability at least $1-\delta/4$, for any $k$-center set $C\in \calC$,
	\begin{align}
	\label{eq:Q1} 
	\sum_{x\in X} \left(d^2(\pi(x),C)+d^2(x,\pi(x))\right)^{z/2}\in (1\pm \eps)\cdot \cost_z(X,C)
	\end{align}
	by Lemma~\ref{lm:projection}, and
	\begin{align}
	\label{eq:coreset_preserve}
	\sum_{x\in D_1} u(x)\cdot d^z(x,C)\in (1\pm \eps) \sum_{x\in D_1} u(x)\cdot \left(d^2(\pi(x),C)+d^2(x,\pi(x))\right)^{z/2}.
	\end{align}
	by Lemma~\ref{lm:preserve}.
	Then we have the following claim.
	
	\begin{claim}
		\label{claim:alg}
		Both $X$ and $D_1$ satisfy the $2\eps$-representativeness property with respect to $\Gamma$.
	\end{claim}
	
	\begin{proof}
		For any equivalence class $\Delta_C^{\Gamma}$ and any two $k$-center sets $C_1,C_2\in \Delta_C^{\Gamma}$, we have
		\begin{eqnarray*}
			\begin{split}
				\cost_z(X,C_1) & \in && (1\pm \eps)	\cdot \sum_{x\in X} \left(d^2(\pi(x),C_1)+d^2(x,\pi(x))\right)^{z/2} && (\text{Ineq.~\eqref{eq:Q1}}) \\
				& \in && (1\pm \eps) \cdot \sum_{x\in X} \left(d^2(\pi(x),C_2)+d^2(x,\pi(x))\right)^{z/2} && (\text{Defition~\ref{def:equivalent}}) \\
				& \in && (1\pm 2\eps)\cdot \cost_z(X,C_2). && (\text{Ineq.~\eqref{eq:Q1}})
			\end{split}
		\end{eqnarray*}
		By the same argument, Inequality~\eqref{eq:coreset_preserve} implies that $D_1$ also satisfies the $\eps$-representativeness property.
		This completes the proof.
	\end{proof}
	
	Now we are ready to prove the theorem.
	Given a $k$-center set $C\in \calC$, suppose $C$ belongs to the equivalence class $\Delta^\Gamma_{C'}$ for some $C'\in \Gamma'$.
	\begin{eqnarray*}
		\begin{split}
			& && \sum_{x\in D_1} u(x)\cdot d^z(x,C) &&\\
			&\in && (1\pm 2\eps)\cdot \sum_{x\in D_1} u(x)\cdot d^z(x,C') && (\text{Claim~\ref{claim:alg}}) \\
			&\in &&(1\pm 2\eps) \cdot \cost_z(X,C') && (\text{Ineq.~\eqref{eq:Q2}}) \\
			&\in && (1\pm 4\eps) \cdot \cost_z(X,C), && (\text{Claim~\ref{claim:alg}})
		\end{split}
	\end{eqnarray*}
	which completes the proof of correctness by letting $\eps' = O(\eps)$.

	For the running time, it costs $O\left(ndk+nd \log(n/\delta) + k^2 \log^2 n + \log^2(1/\delta) \log^2 n\right)$ time to construct a $2^{O(z)}$-approximate solution $C^\star$~\cite{mettu2004optimal},\footnote{\cite{mettu2004optimal} only discuss $k$-median, but their construction can be easily generalized to \kzC by the relaxed triangle inequality of $d^z$.} $O(ndk)$ time to compute all $X_c$ and $\sigma_1(x)$, and $O(N_1)=O(n)$ time to construct $D_1$.
	Hence, we prove for the overall running time.
\end{proof}

\subsection{Analyzing the second stage of Algorithm~\ref{alg:coreset}}
\label{sec:second_stage}

Next, we prove the reduction theorem (Theorem~\ref{thm:reduction}) that provides a theoretical guarantee for the second stage.
The main idea is to apply the result on terminal embeddings such that the dimension is further reduced to $O\left(\eps^{-2} \log(k/\eps)\right)$.

\begin{proof}
	%
	It costs $O(N_1 dk)=O(ndk)$ time to compute all $D_c$ and $\sigma_2(x)$, and $O(N_2)=O(n)$ time to construct $S$.
	Hence, we only need to focus on the correctness. 
	Since we suppose that the ouput $D_1$ of the first stage is an $\eps$-coreset over $X$, we have that
	\[
	\sum_{x\in D_1} u(x)\cdot d^z(x,C^\star)\leq (1+\eps)\cdot \cost_z(X,C^\star) \leq 2\alpha\cdot \OPT_z.
	\]
	Hence, $C^\star$ is also an $O(1)$-approximation of the \kzC problem over $D_1$.
	Let $f:\R^d\rightarrow \R^m$ be a terminal embedding of $D_1$ where $m=O(z^2 \eps^{-2}\log N_1)$.
	By Theorem~\ref{thm:embedding}, we have that for any $x\in D_1$ and $y\in \R^d$,
	\begin{eqnarray}
	\label{eq:embedding}
	d^z(x,y)\leq d^z\left(f(x),f(y)\right)\leq \left(1+\frac{\eps}{10z}\right)^z\cdot d^z(x,y)\leq (1+\eps)\cdot d^z(x,y).
	\end{eqnarray}
	Hence, we have for any set $A\subseteq \R^d$,
	\begin{eqnarray}
	\label{eq:set}
	d^z(x,A)\leq d^z\left(f(x),f(A)\right)\leq (1+\eps)\cdot d^z(x,A).
	\end{eqnarray}
	Then $f(C^\star)$ is an $O(1)$-approximation of the \kzC problem over the weighted point set $f(D_1)$ with weights $u(x)$.
	By Theorem~\ref{thm:fl11_bfl16}, with probability at least $1-\delta/2$, $f(S)$ together with weights $w(x)$ is an $\eps$-coreset for \kzC over $(f(D_1),u)$ since $N_2= \Omega\left(\eps^{-2z} (km \log k+\log(1/\delta)\right)$.
	Then it suffices to prove that $S$ together with weights $w(x)$ is an $O(\eps)$-coreset for \kzC over $X$.
	For any $k$-center set $C\in \calC$, we have the following
	\begin{enumerate}
		\item[P1.] $\sum_{x\in D_1} u(x)\cdot d^z(x,C)\in (1\pm \eps)\cdot \cost_z(X,C)$ by the assumption of the theorem.
		\item[P2.] $\sum_{x\in S} w(x)\cdot d^z\left(f(x),f(C)\right)\in (1\pm \eps)\cdot \sum_{x\in D_1} u(x)\cdot d^z\left(f(x),f(C)\right)$ by the definition of $S$ and Theorem~\ref{thm:fl11_bfl16}.
		\item[P3.] $\sum_{x\in D_1} u(x)\cdot d^z\left(f(x),f(C)\right)\in (1\pm \eps)\cdot \sum_{x\in D_1} u(x)\cdot d^z(x,C)$ by Inequality~\eqref{eq:set}.
		\item[P4.] $\sum_{x\in S} w(x)\cdot d^z(x,C) \in (1\pm \eps)\cdot \sum_{x\in S} w(x)\cdot d^z\left(f(x),f(C)\right)$ by Inequality~\eqref{eq:set}.
	\end{enumerate}
	Combining the above properties, we have that
	\begin{eqnarray*}
		\begin{split}
			\sum_{x\in S} w(x)\cdot d^z(x,C)  & \in &&  (1\pm \eps)\cdot \sum_{x\in S} w(x)\cdot d^z(f(x),f(C)) &&   (\text{P4}) \\
			&  \in && (1\pm 2\eps)\cdot \sum_{x\in D_1} u(x)\cdot d^z(f(x),f(C)) &&  (\text{P2}) \\
			&   \in && (1\pm 3\eps)\cdot \sum_{x\in D_1} u(x)\cdot d^z(x,C)  &&  (\text{P3}) \\
			&  \in && (1\pm 4\eps)\cdot \cost_z(X,C). && (\text{P1}) 
		\end{split}
	\end{eqnarray*}
	which completes the proof.
\end{proof}

\begin{remark}
	\label{remark:another_alg}
	In the second stage of Algorithm~\ref{alg:coreset}, we apply the first framework stated in Theorem~\ref{thm:fl11_bfl16}.
	This is because we want to reduce the dependence of size on $k$ to be linear.
	In the case that $\eps$ is small, we can apply the second framework stated in Theorem~\ref{thm:fl11_bfl16} instead. 
	By Theorem~\ref{thm:fl11_bfl16}, the coreset size should be $O\left((2^{2z} \eps^{-2} k\cdot (km \log k + \log(1/\delta)) \right)$ where $m=O(\eps^{-2} \log(N_1/\eps))$ by Theorem~\ref{thm:reduction}.
	This provides us an $\eps$-coreset of size 
	\[
	O(2^{2z} \eps^{-4} k^2\log(k/\eps) \log (k/\eps \delta)).
	\]
\end{remark}

\subsection{Proof of the main technical Lemma~\ref{lm:preserve}}
\label{sec:proof_mainlm}

For preparation, we introduce the following theorem showing the existence of a weak-coreset $S$ for $(k,z)$-subspace approximation over $X$ of size independent of $n,d$.
Recall that $\calP$ is the collection of all $j$-flats in $\R^d$ with $j\leq k$, i.e., all subspaces in $\R^d$ of dimension at most $k$.

\begin{theorem}[\bf{Weak-coreset for subspace approximation}]
	\label{thm:subspace}
	Given a dataset $X$ of $n$ points in $\R^d$, $\eps,\delta \in (0,0.5)$, constant $z\geq 1$ and integer $k \geq 1$, suppose $\sigma: X\rightarrow \R_{\geq 0}$ is a sensitivity function satisfying that
	\[
	\sigma(x) \geq \sup_{P\subseteq \calP} \frac{d^z(x, P)}{\sum_{y\in X} d^z(y, P)}
	\] 
	for each $x\in X$.
	Let $\calG = \sum_{x\in X} \sigma(x)$ denote the total sensitivity.
	Suppose $S\subseteq X$ is constructed by taking 
	\[
	O\left(\frac{\calG^2}{\eps^2}\cdot (\eps^{-1} k^3\log (k/\eps)+\log(1/\delta))\right)
	\] 
	samples, where each sample $x\in X$ is selected with probability $\frac{\sigma(x)}{\calG}$ and has weight $w(x):= \frac{\calG}{|S|\cdot \sigma(x)}$.
	Then with probability at least $1-\delta$, $S$ is an $\eps$-weak-coreset for the $(k,z)$-subspace approximation problem over $X$.
\end{theorem}

\noindent
Actually, the above construction implies an algorithm to compute a nearly optimal solution for the $(k,z)$-subspace approximation problem over $X$; see discussion in Remark~\ref{remark:weak_coreset}.
To prove the theorem, we need the following lemma based on~\cite[Theorem 9]{deshpande2007sampling}.
It indicates that a nearly optimal solution for $(k,z)$-subspace approximation exists in some low dimensional space.

\begin{lemma}[\bf{Existence of approximate $k$-flats in low dimensional subspaces}]
	\label{lm:projective_clustering}
	Given a weighted dataset $X$ of $n$ points together with weights $u(x)$ in $\R^d$, $ \eps \in (0,0.5)$, constant $z\geq 1$ and integer $k \geq 1$, there exists a $k$-flat $P$ that is spanned by at most $O(\eps^{-1} k^2 \log(k/\eps))$ points from $X$, such that 
	\[
	\sum_{x\in X} u(x)\cdot d^z(x,P)\leq (1+\eps)\cdot \min_{P'\in \calP}\sum_{x\in X} u(x)\cdot d^z(x,P').
	\]
\end{lemma}

\begin{proof}
	By~\cite[Theorem 1.3]{shyamalkumar2007efficient}, there exists a collection $D\subseteq X$ of $O(\eps^{-1} k \log(1/\eps)$ points such that the spanned subspace of $D$ contains a $k$-flat $P$ satisfying that
	\[
	\sum_{x\in X} u(x)\cdot d^z(x,P)\leq (1+\eps)^{(k+1)z}\cdot \min_{P'\in \calP}\sum_{x\in X} u(x)\cdot d^z(x,P').
	\]
	Replacing $\eps'=O(\eps/zk)$, we complete the proof.
\end{proof}

\noindent
We are ready to prove the theorem.

\begin{proof}[Proof of Theorem~\ref{thm:subspace}]
	Denote $\calP'$ to be the collection of all $k$-flats that are spanned by at most $O\left(\eps^{-1} k^2 \log(k/\eps)\right)$ points from $X$.
	By~\cite[Lemma 8.2]{feldman2011unified}, the function dimension\footnote{Since this paper only uses function dimension as a black box, we do not present the definition. We refer interested readers to~\cite[Definition 6.4]{feldman2011unified} or~\cite[Definition 4.5]{braverman2016new} for concrete definitions.} of $(X,\calP')$ is $O\left(\eps^{-1} k^3 \log(k/\eps)\right)$.
	Then by~\cite[Theorem 4]{varadarajan2012sensitivity}, with probability at least $1-\delta$, for any $k$-flat $P\in \calP'$, 
	\begin{align}
	\label{eq:vx12_flat}
	\sum_{x\in S} w(x)\cdot d^z(x,P)\in (1\pm \eps)\cdot\sum_{x\in X} d^z(x,P).
	\end{align}
	Then we have
	\begin{eqnarray*}
		\begin{split}
			\min_{P\in \calP} \sum_{x\in S} w(x)\cdot d^z(x,P) &\geq &&	(1-\eps)\cdot\min_{P\in \calP'} \sum_{x\in S} w(x)\cdot d^z(x,P)&& (\text{Lemma~\ref{lm:projective_clustering}})\\
			& \geq && (1-\eps)^2\cdot\min_{P\in \calP'} \sum_{x\in X} d^z(x,P) && (\text{Ineq.~\eqref{eq:vx12_flat}}) \\
			&\geq &&(1-\eps)^2\cdot\min_{P\in \calP} \sum_{x\in X} d^z(x,P). &&
		\end{split}
	\end{eqnarray*}
	We also have
	\begin{eqnarray*}
		\begin{split}
			\min_{P\in \calP} \sum_{x\in S} w(x)\cdot d^z(x,P) &\leq &&	\min_{P\in \calP'} \sum_{x\in S} w(x)\cdot d^z(x,P)&& \\
			&\leq &&(1+\eps)\cdot\min_{P\in \calP'} \sum_{x\in X} d^z(x,P) &&  (\text{Ineq.~\eqref{eq:vx12_flat}}) \\
			& \leq && (1+\eps)^2\cdot \min_{P\in \calP}\sum_{x\in X} w(x)\cdot d^z(x,P). && (\text{Lemma~\ref{lm:projective_clustering}})
		\end{split}
	\end{eqnarray*}

	Letting $\eps'=O(\eps)$, we complete the proof.
\end{proof}

\begin{remark}
	\label{remark:weak_coreset}
	Theorem~\ref{thm:subspace} actually provides an approach to compute a $(1+\eps)$-approximate solution for the $(k,z)$-subspace approximation problem.
	Suppose $P^\star\in \calP'$ is a $k$-flat satisfying that
	\[
	\sum_{x\in S} w(x)\cdot d^z(x,P^\star)\leq (1+\eps)\cdot \min_{P\in \calP'} \sum_{x\in S} w(x)\cdot d^z(x,P).
	\]
	Then by the above proof, we directly have
	\[
	(1-\eps)\cdot \min_{P\in \calP} \sum_{x\in X} d^z(x,P)\leq \sum_{x\in S} w(x)\cdot d^z(x,P^\star)\leq (1+\eps)^3\cdot \min_{P\in \calP} \sum_{x\in X} d^z(x,P),
	\]
	which indicates that $\sum_{x\in S} w(x)\cdot d^z(x,P^\star)$ is a $(1\pm O(\eps))$-approximation of the $(k,z)$-subspace approximation objective $\min_{P\in \calP} \sum_{x\in X} d^z(x,P)$.
	Moreover, since $P^\star\in \calP'$, we also have that
	\[
	\sum_{x\in S} w(x)\cdot d^z(x,P^\star) \in (1\pm \eps) \sum_{x\in X}  d^z(x,P^\star)
	\]
	by Inequality~\eqref{eq:vx12_flat}.
	Thus, $P^\star$ is a $(1+O(\eps))$-approximate solution for the $(k,z)$-subspace approximation problem.
\end{remark}

\noindent
Now we can prove the main lemma.

\begin{proof}[Proof of Lemma~\ref{lm:preserve}]
	Let $\eps' = \frac{ \eps^{z+3}}{6\cdot (84z)^{2z}}$.
	Recall that we have $\Gamma$ is a subspace of $\R^d$ satisfying that $C^\star\subset \Gamma$ and for any $k$-center set $C\in \calC$, 
	\[
	\sum_{x\in X} \left(d^z(x,\pi(x))-d^z(x,\pi_C(x))\right)= \frac{\eps'}{2} \cdot \OPT_z,
	\]
	We first have the following observations
	\begin{eqnarray}
	\label{eq:sampling1} 
	\begin{split}
	\sum_{x\in X} d^z\left(x,\pi(x)\right) &\leq&& \sum_{x\in X} d^z(x,C^\star) && (C^\star\in \Gamma)\\
	&\leq &&\alpha\cdot \OPT_z, && (\text{Defn. of $C^\star$})
	\end{split}
	\end{eqnarray}
	and
	\begin{eqnarray}
	\label{eq:sampling2} 
	\begin{split}
	\sigma_1(x)&>&& 2^{2z+2} \alpha^2 \cdot \frac{d^z(x,c^\star(x))}{ \cost_z(X,C^\star)} && (\text{Defn. of $\sigma_1(x)$})\\
	&\geq && \frac{2^{2z+2} \alpha\cdot d^z(x,\pi(x))}{\OPT_z}, && (C^\star\in \Gamma)
	\end{split}
	\end{eqnarray}
	and
	\begin{eqnarray}
	\label{eq:sampling3}
	\begin{split}
	\sum_{x\in X}\sigma_1(x) &= &&2^{2z+2} \alpha^2\cdot \sum_{x\in X} \left(\frac{d^z(x,c^\star(x))}{ \cost_z(X,C^\star)}+\frac{1}{|X_{c^\star(x)}|}\right) &&\\
	&\leq && 2^{2z+2} \alpha^2\cdot(1+k) && (|C^\star|=k) \\
	&\leq && 2^{2z+3} \alpha^2k. &&
	\end{split}
	\end{eqnarray}

	For a $k$-center set $C\in \calC$, recall that $\pi_C$ is the projection from $X$ to $\mathrm{Span}(\Gamma\cup C)$.
	We claim that
	\begin{align}
	\label{eq:case1}
	\min_{C\in \calC}\sum_{x\in D_1} u(x)\cdot d^z(x,\pi_C(x))\geq \min_{C\in \calC}\sum_{x\in X} d^z(x,\pi_C(x))-\frac{\eps'\cdot \OPT_z}{2}.
	\end{align}
	Let $\widehat{C}\in \calC$ denote the $k$-center set such that $\sum_{x\in X} d^z(x,\pi_C(x))$ is minimized.
	To prove this inequality, we consider two cases.
	If $\sum_{x\in X} d^z(x,\pi_{\widehat{C}}(x))\leq \frac{\eps'\cdot \OPT_z}{2}$, we directly have
	\[
	\min_{C\in \calC}\sum_{x\in D_1} u(x)\cdot d^z(x,\pi_C(x))\geq 0\geq \min_{C\in \calC}\sum_{x\in X} d^z(x,\pi_C(x))-\frac{\eps'\cdot \OPT_z}{2}.
	\]
	Otherwise, suppose $\sum_{x\in X} d^z(x,\pi_{\widehat{C}}(x))> \eps'\cdot \OPT_z/2$.
	Since $C^\star\subseteq \Gamma$, we have that for any $k$-center set $C\in \calC$, $\sum_{x\in X} d^z(x,\pi_C(x))\leq \cost_z(X,C^\star)$.
	We regard $X$ as a point set in $\Gamma^\perp$ (i.e., the orthogonal complement of $\Gamma$).
	Then each $k$-center set $C\in \calC$ corresponds to a subspace $H\subseteq \Gamma^\perp$ of dimension at most $k$, satisfying that $\mathrm{Span}(\Gamma\cup C) =\mathrm{Span}(\Gamma\cup H)$. 
	This enables us to apply Theorem~\ref{thm:subspace} to $\Gamma^\perp$.
	We set $\sigma(x)$ in Theorem~\ref{thm:subspace} as follows:
	\begin{eqnarray*}
		\begin{split}
			\sigma(x) &:= &&\frac{\sigma_1(x)}{2^{2z+2}\alpha^2} \cdot \frac{\cost_z(X,C^\star)}{\sum_{x\in X} d^z(x,\pi_{\widehat{C}}(x))} &&\\
			&= && \frac{d^z(x,c^\star(x))}{\sum_{x\in X} d^z(x,\pi_{\widehat{C}}(x))}  +\frac{\cost_z(X,C^\star)}{|X_{c^\star(x)}|\cdot \left(\sum_{x\in X} d^z(x,\pi_{\widehat{C}}(x))\right)} && (\text{Defn. of $\sigma_1(x)$})\\
			&\geq && \sup_{C\in \calC} \frac{d^z(x,\pi_C(x))}{\sum_{x\in X} d^z(x,\pi_C(x))}. && (C^\star\in \Gamma\text{ and Defn. of $\widehat{C}$})
		\end{split}
	\end{eqnarray*}
	Note that the sampling distribution with respect to $\sigma$ is exactly the same as to $\sigma_1$.
	Moreover, we have
	\begin{eqnarray*}
		\begin{split}
			\calG &:=&&\sum_{x\in X} \sigma(x) &&\\
			&=&& \sum_{x\in X} \frac{d^z(x,c^\star(x))}{\sum_{x\in X} d^z(x,\pi_{\widehat{C}}(x))}  +\frac{\cost_z(X,C^\star)}{|X_{c^\star(x)}|\cdot \left(\sum_{x\in X} d^z(x,\pi_{\widehat{C}}(x))\right)} &&\\
			&= &&  \frac{(k+1)\cdot \cost_z(X,C^\star)}{\sum_{x\in X} d^z(x,\pi_{\widehat{C}}(x))} && (|C^\star|=k)\\
			&\leq && \frac{\alpha(k+1)\cdot \OPT_z}{\eps'\cdot \OPT_z/2} && (\sum_{x\in X} d^z(x,\pi_{\widehat{C}}(x))> \frac{\eps'\cdot \OPT_z}{2})\\
			&=&&\frac{2\alpha (k+1)}{\eps'}. &&
		\end{split}
	\end{eqnarray*}
	Hence, $N_1 = \Omega\left(\frac{\calG^2}{(\eps')^2}\cdot ((\eps')^{-1} k^3\log \frac{k}{\eps'}+\log\frac{1}{\delta})\right)$ as stated in Theorem~\ref{thm:subspace}.
	By Theorem~\ref{thm:subspace}, we have that with probability at least $1-\delta/8$,
	\begin{eqnarray*}
		\begin{split}
			\min_{C\in \calC}\sum_{x\in D_1} u(x)\cdot d^z(x,\pi_C(x)) &	\geq && (1-\frac{\eps'}{2\alpha})\cdot \min_{C\in \calC}\sum_{x\in X} d^z(x,\pi_C(x)) && \\
			&\geq && \min_{C\in \calC}\sum_{x\in X} d^z(x,\pi_C(x))-\frac{\eps'}{2\alpha}\cdot \cost_z(X,C^\star). && \\
			&\geq &&\min_{C\in \calC}\sum_{x\in X} d^z(x,\pi_C(x))-\frac{\eps'\cdot \OPT_z}{2}, && 
		\end{split}
	\end{eqnarray*}
	which completes the proof of Inequality~\eqref{eq:case1}.

	Next, we prove that with probability at least $1-\delta/8$, the following property holds:
	\begin{align}
	\label{eq:case1_2}
	\sum_{x\in D_1} u(x)\cdot d^z(x,\pi(x)) \leq \sum_{x\in X} d^z(x,\pi(x)) + \frac{\eps' \cdot \OPT_z}{2}.
	\end{align}
	For each sample $x\in D_1$, we note that
	\begin{eqnarray}
	\label{eq:variance}
	\begin{split}
	|D_1|\cdot u(x)\cdot d^z(x,\pi(x))& = && \frac{\sum_{y\in X}\sigma_1(y)}{\sigma_1(x)}\cdot d^z(x,\pi(x)) && \\
	&\geq &&\frac{2^{2z+3} \alpha^2k}{\frac{2^{2z+2} \alpha\cdot d^z(x,\pi(x))}{\OPT_z}}\cdot d^z(x,\pi(x)) && (\text{Ineqs.~\eqref{eq:sampling2} and~\eqref{eq:sampling3}}) \\
	&= && 2\alpha k \cdot \OPT_z. &&
	\end{split}
	\end{eqnarray}
	Then by Hoeffding's inequality, we have that
	\begin{align*}
	&	&& \Pr\left[\left|\sum_{x\in X} d^z(x,\pi(x))-\sum_{x\in D_1} u(x)\cdot d^z(x,\pi(x))\right|\geq \frac{\eps'\cdot \OPT_z}{2}\right] &&\\
	&\leq && 2\cdot \exp\left(-\frac{2(\frac{\eps'\cdot \OPT_z}{2})^2}{N_1\cdot (2\alpha k \cdot \OPT_z)^2}\right)&& (\text{Ineq.~\eqref{eq:variance}})\\ 
	&\leq &&\frac{\delta}{8}, && (\text{value of $N_1$})
	\end{align*}
	which completes the proof of Inequality~\eqref{eq:case1_2}.

	Now we are ready to prove the lemma. 
	With probability at least $1-\delta/4$, Inequalities~\eqref{eq:case1} and~\eqref{eq:case1_2} hold (union bound). 
	Then for any $k$-center set $C\in \calC$,
	\begin{eqnarray*}
		\begin{split}
			&	&& \sum_{x\in D_1} u(x)\cdot d^z(x,\pi(x)) - \sum_{x\in D_1} u(x)\cdot d^z(x,\pi_C(x)) &&\\
			&\leq && \sum_{x\in D_1} u(x)\cdot d^z(x,\pi(x)) - \min_{C'\in \calC}\sum_{x\in D_1} u(x)\cdot d^z(x,\pi_{C'}(x)) && \\
			&\leq && \sum_{x\in X} d^z(x,\pi(x))+\frac{\eps'\cdot \OPT_z}{2} &&\\
			& &&-\min_{C'\in \calC}\sum_{x\in X} d^z(x,\pi_{C'}(x)) + \frac{\eps' \cdot \OPT_z}{2} && (\text{Ineqs.~\eqref{eq:case1} and~\eqref{eq:case1_2}})\\
			&\leq &&\sum_{x\in X} d^z(x,\pi(x)) - \sum_{x\in X} d^z(x,\pi_C(x)) + \eps'\cdot \OPT_z &&  \\
			&\leq && 2\eps' \cdot \OPT_z.&& (\text{by assumption})
		\end{split}
	\end{eqnarray*}
	By Lemma~\ref{lm:projection}, we complete the proof of Lemma~\ref{lm:preserve}.
\end{proof}

\subsection{Geometric observations}
\label{sec:discussion}

Note that the first stage of Algorithm~\ref{alg:coreset} is almost the same to the second framework stated in Theorem~\ref{thm:fl11_bfl16} except that the coreset size $N_1$ is independent of $d$. 
In this section, we discuss the geometric observations that makes $N_1$ samples enough for an $\eps$-coreset.

Construct a subspace $\Gamma\subsetneq \R^d$ of dimension $\poly(k/\eps)$ by~\cite[Algorithm 1]{sohler2018strong}, which leads to Inequality~\eqref{eq:Q1} by Lemma~\ref{lm:projection}.
Recall that $\Gamma'$ is obtained from $\Gamma$ by appending an arbitrary dimension in $\R^d$ that is orthogonal to $\Gamma$.
Also recall that $\calC_\Gamma$ denotes the collection of $k$-center sets $C\subset \Gamma'$.
We have the following geometric observations implying that we only need to approximately preserve all \kzC objectives with respect to $k$-center sets in $\Gamma'$ instead of the whole $\calC$.
This reduces the function dimension of $k$-center sets from $O(dk)$ to $\poly(k/\eps)$.
The first observation follows from Claim~\ref{claim:alg}.

\begin{observation}[\bf{Representativeness property for $X$}]
	\label{ob:equivalent}
	$X$ satisfies the $\eps$-representativeness property with respect to $\Gamma$.
\end{observation}

\noindent
Moreover, the representativeness property can be generalized to subsets of $X$ that are weak-coresets for the $(k,z)$-subspace approximation problem.

\begin{observation}[\bf{Representativeness property for weighted subsets of $X$}]
	\label{ob:equivalent_subset}
	Let $S$ be a weighted subset of $X$ together with a weight function $w: S\rightarrow \R_{\geq 0}$ and $\eps' = \frac{ \eps^{z+3}}{6\cdot (84z)^{2z}}$.
	$S$ satisfies the $\eps$-representativeness property with respect to $\Gamma$ if the following holds:
	\begin{enumerate}
		\item $S$ is an $\eps'$-weak-coreset for the $(k,z)$-subspace approximation problem in $\Gamma^\perp$.
		\item $S$ approximately preserves the $l_z$-subspace approximation objective with respect to $\Gamma$, i.e., 
		\[
		\sum_{x\in S} w(x)\cdot d^z(x,\Gamma) \in \sum_{x\in X} d^z(x,\Gamma)\pm \eps'\cdot \OPT_z.
		\]
	\end{enumerate}
\end{observation}

\begin{proof}
	By the proof of Lemma~\ref{lm:preserve}, these two conditions imply that for any $k$-center set $C\in \calC$,
	\[
	\sum_{x\in S} w(x)\cdot d^z(x,\pi(x)) - \sum_{x\in S} w(x)\cdot d^z(x,\pi_C(x))\leq O(\eps')\cdot \OPT_z,
	\]
	where $\pi$ and $\pi_C$ denote the projection from $X$ to $\Gamma$ and $\mathrm{Span}(\Gamma\cup C)$ respectively.
	Then by Lemma~\ref{lm:preserve}, Inequality~\eqref{eq:coreset_preserve} holds.
	By Claim~\ref{claim:alg}, we complete the proof.
\end{proof}

\noindent
Now suppose we have a weighted subset $S\subseteq X$ that satisfies the $\eps$-representativeness property.
By Definition~\ref{def:representativeness}, if $S$ approximately preserves the \kzC objective for some $k$-center set $C\in \calC_\Gamma$ over $X$, then we directly have that $S$ approximately preserves all \kzC objectives with respect to $k$-center sets within the whole equivalence class $\Delta^\Gamma_C$.
Hence, we only need to consider those $k$-center sets in $\Gamma$ instead of $\R^d$ and conclude the following corollary. 
The corollary indicates that coreset for clustering in low dimensional subspace plus weak-coreset for subspace approximation implies coreset for clustering in $\R^d$.

\begin{corollary}[\bf{Dimension reduction for \kzC}]
	\label{cor:dimension_reduction_clustering}
	For every dataset $X$ of $n$ points in $\R^d$, $\eps,\delta \in (0,0.5)$, constant $z\geq 1$ and integer $k \geq 1$, let $\eps' = \frac{ \eps^{z+3}}{6\cdot (84z)^{2z}}$.
	There exists a subspace $\Gamma\subsetneq \R^d$ of dimension $O(k/\eps')$ such that for any weighted point set $S\subseteq X$ together with a weight function $w: S\rightarrow \R_{\geq 0}$, $S$ is an $O(\eps)$-coreset for \kzC if 
	\begin{enumerate}
		\item $S$ is an $\eps$-coreset for \kzC in subspace $\Gamma'$; 
		\item $S$ is an $\eps'$-weak-coreset for the $(k,z)$-subspace approximation problem in $\Gamma^\perp$.
		\item $S$ approximately preserves the $l_z$-subspace approximation objective with respect to $\Gamma$, i.e., 
		\[
		\sum_{x\in S} w(x)\cdot d^z(x,\Gamma) \in \sum_{x\in X} d^z(x,\Gamma)\pm \eps'\cdot \OPT_z.
		\]
	\end{enumerate}
\end{corollary}

\noindent
In fact, the above corollary can be generalized to other shape fitting problems.
The main reason is that Lemma~\ref{lm:projection} not only holds for $k$-center sets but also holds for an arbitrary non-empty set that is contained in a $k$-dimensional subspace by~\cite[Theorem 10]{sohler2018strong}.
For instance, if we consider $\calP$ that is the collection of all $j$-flats ($j\leq k$), then Corollary~\ref{cor:dimension_reduction_clustering} can be translated to a dimension reduction result for subspace approximation as follows.

\begin{corollary}[\bf{Dimension reduction for subspace approximation}]
	\label{cor:dimension_reduction_subspace}
	For every dataset $X$ of $n$ points in $\R^d$, $\eps,\delta \in (0,0.5)$, constant $z\geq 1$ and integer $k \geq 1$, let $\eps' = \frac{ \eps^{z+3}}{6\cdot (84z)^{2z}}$.
	There exists a subspace $\Gamma\subsetneq \R^d$ of dimension $O(k/\eps')$ such that for any weighted point set $S\subseteq X$ together with a weight function $w: S\rightarrow \R_{\geq 0}$, $S$ is an $O(\eps)$-coreset for $(k,z)$-subspace approximation if 
	\begin{enumerate}
		\item $S$ is an $\eps$-coreset for $(k,z)$-subspace approximation in subspace $\Gamma'$; 
		\item $S$ is an $\eps'$-weak-coreset for the $(k,z)$-subspace approximation problem in $\Gamma^\perp$.
		\item $S$ approximately preserves the $l_z$-subspace approximation objective with respect to $\Gamma$, i.e., 
		\[
		\sum_{x\in S} w(x)\cdot d^z(x,\Gamma) \in \sum_{x\in X} d^z(x,\Gamma)\pm \min_{P\in \calP} \sum_{x\in X} d^z(x, P).
		\] 
	\end{enumerate}
\end{corollary}

\noindent
Similarly, by the \FL framework, this corollary provides an $\eps$-coreset for $(k,z)$-subspace approximation of size $\poly(k/\eps)$, which matches the result in~\cite{sohler2018strong}.
Moreover, the coreset size can be further decreased by applying terminal embedding similar to Theorem~\ref{thm:reduction}.

\subsection{Generalization of Theorem~\ref{thm:coreset} to $\ell_p$-metrics}
\label{sec:generalization}

Given $p\geq 1$, $\ell_p$-metric is induced by distance function $d_p: \R^d\times \R^d\rightarrow \R_{\geq 0}$, where for any two points $x,y\in \R^d$,
\begin{align}
\label{eq:lp_distance}
d_p(x,y) := \left(\sum_{i\in [d]} |x_i-y_i|^p \right)^{1/p}.
\end{align}
The formulation captures classic distances, including Manhattan distance (where $p=1$), Euclidean distance (where $p=2$) and Chebyshev distance (where $p=\infty$).
With $\ell_p$-metric, the \kzC objective with respect to some $C\in \calC$ is defined as follows
\[
\cost_{p,z}(X, C) := \sum_{x \in X}{d_p^z(x, C)},
\]
where, throughout, $d_p^z$ denotes $\ell_p^d$-distance raised to power $z\ge 1$,
and 
\[
d_p(x, C):=\min\left\{d_p(x,c): c\in C\right\}.
\]
We can generalize Definition~\ref{def:coreset} to $\ell_p$-metrics.
\begin{definition}[\bf{Coresets for \kzC with $\ell_p$-metric in $\R^d$}]
	\label{def:coreset_lp}
	Given a collection $X\subseteq \R^d$ of $n$ weighted points and $\eps\in (0,1)$, an $\eps$-coreset for \kzC in $\ell_p^d$ metric spaces is a subset $S \subseteq \R^d$ with weights $w : S \rightarrow \R_{\geq 0}$ such that for any $k$-center set $C\in \calC$, the \kzC objective with respect to $C$ is $\eps$-approximately preserved, i.e.,
	\begin{equation*}
	\sum_{x \in S}w(x) \cdot d_p^z(x, C)
	\in (1 \pm \eps) \cdot \cost_{p,z}(X, C).
	\end{equation*}
\end{definition}

\noindent
Note that Theorem~\ref{thm:coreset} considers the Euclidean distance where $p=2$ and we want to generalize Theorem~\ref{thm:coreset} to all $p\geq 1$.
In this section, we show that Theorem~\ref{thm:coreset} can be generalized to $\ell_p$-metrics for $1\leq p\leq 2$; see the following corollary.
The main idea is that for $1\leq p < 2$, there exists an isometric embedding from $\ell_p$ to $\ell_2$ square~\cite{kahane1993some}. 
By this idea, we can reduce the problem of constructing an $\eps$-coreset for \kzC with $\ell_p$-metric to constructing an $O(\eps)$-coreset for \ProblemName{$(k, 2z)$-Clustering} with $\ell_2$-metric.

\begin{corollary}[\bf{Coresets for \kzC with $\ell_p$-metrics ($1\leq p < 2$)}]
	\label{cor:coreset_lp}
	There exists a randomized algorithm that, for a given dataset $X$ of $n$ points in $\R^d$, integer $k \geq 1$, $1\leq p<2$, constant $z\geq 2$ and $\eps \in (0,0.5)$, with probability at least $1-\delta$, constructs an $\eps$-coreset for \kzC with $\ell_p$-metric of size  
	\[
	O\left(\min\left\{\eps^{-4z-2}, 2^{4z} \eps^{-4} k \right\} k\log k \log\frac{k}{\eps\delta}\right)
	\]
	and runs in time 
	\[
	O\left(ndk+nd \log(n/\delta) + k^2 \log^2 n + \log^2(1/\delta) \log^2 n\right).
	\]
\end{corollary}

\begin{proof}
	Let $\OPT_{p,z}$ denote the optimal \kzC objective of $X$ with $\ell_p$-metric.
	Let $\diam_p(X) = \max_{x,x'\in X} d_p(x,x')$ denote the $\ell_p$-diameter of $X$. 
	Let 
	\[
	H = \left\{y\in \R^d: d_p(y,X)\leq \frac{20z\cdot \diam_p(X)}{\eps}\right\}
	\] 
	denote the collection of points whose $\ell_p$-distance to $X$ is at most $\frac{20}{\eps}$ times the $\ell_p$-diameter of $X$.
	Given $\tau$ with
	\[
	0< \tau \leq \frac{\eps}{20^z}\cdot \left(\frac{\eps \OPT_{p,z}}{n} \right)^{1/z},
	\]
	let $H_\tau$ denote the combination of $X$ and a $\tau$-net of $H$.\footnote{A $\tau$-net $Q$ means that for any point $x\in H$, there exists a point $q\in Q$ such that $d(x,q)\leq \tau$.}
	By~\cite{kahane1993some}, let $f: H_\tau\rightarrow \R^m$ ($1\leq m\leq |H_\tau|$ is an integer) denote an isometric embedding from $\ell_p$ to $\ell_2$ square satisfying that for any $x,y\in H_\tau$,
	\[
	d_p(x,y) = d_2^2(f(x),f(y)).
	\]
	By definition, we know that $f$ is a one-to-one mapping.
	Then for any $k$-center set $C\subseteq H_\tau$, we have
	\begin{align}
	\label{eq:cor_embed}
	\cost_{p,z}(X,C) = \cost_{2,2z}(f(X),f(C)).
	\end{align}
	Suppose $S\subseteq f(X)$ together with a weight function $w:S\rightarrow \R_{\geq 0}$ is a $\eps$-coreset for \ProblemName{$(k, 2z)$-Clustering} of $f(X)$ with $\ell_p$-metric in $\R^m$.
	We first have that 
	\begin{align}
	\label{eq:total_weight}
	\sum_{x\in S} w(x)\in (1\pm 2\eps) \cdot n.
	\end{align}
	This is because letting $C = \left\{c\in \R^m\right\}$ for some point $c$ with $d_2(c, X)\rightarrow +\infty$ and $x^o$ be an arbitrary point in $S$, we have
	\begin{eqnarray*}
	\begin{split}
	\sum_{x\in S}w(x) \cdot d_2^{2z}(x^o, c) & \in && (1\pm \frac{\eps}{2})\cdot \sum_{x\in S} w(x)\cdot d_2^{2z} (x,C)  & (d_2^{2z}(x^o, c)\in (1\pm \frac{\eps}{2})\cdot d_2^{2z}(c, S))\\
	&\in && (1\pm \frac{3\eps}{2})\cdot \cost_{2,2z}(f(X),f(C)) & (\text{Definition~\ref{def:coreset_lp}}) \\
	& \in && (1\pm 2\eps)\cdot  n\cdot d_2^{2z}(x^o, c). & (d_2^{2z}(x^o, c)\in (1\pm \frac{\eps}{2})\cdot d_2^{2z}(c, X))
	\end{split}
	\end{eqnarray*}
	Then we consider the weighted subset $f^{-1}(S)$ together with a weight function $w': f^{-1}(S)\rightarrow \R_{\geq 0}$ satisfying that $w'(x) = w(f(x))$ for all $x\in f^{-1}(S)$.
	We claim that for any $C\in \calC$,
	\begin{align}
	\label{eq:cor_claim}
	\sum_{x\in f^{-1}(S)} w'(x)\cdot d_p^z(x,C) \in (1\pm O(\eps))\cdot \cost_{p,z}(X,C).
	\end{align}
	We discuss the following types of $k$-center sets $C\in \calC$.

	\paragraph{Case 1: $C\cap H=\emptyset$.} Let $x^o$ be an arbitrary point in $f^{-1}(S)$. 
	By definition of $H$, we have that for any $c\in C$ and $x\in X$,
	\begin{eqnarray}
	\label{eq:cor1}
	\begin{split}
	d_p^z(x,c) & \in && \left(d_p(x^o,c) \pm d_p(x^o,x)\right)^z & (\text{triangle ineq.}) \\
	& \in && \left(d_p(x^o,c) \pm \diam_p\right)^z & (\text{Defn. of $\diam_p$}) \\
	& \in && (1\pm \eps)\cdot d_p^z(x^o,c). & (\text{Defn. of $H$})
	\end{split}
	\end{eqnarray}
	Hence, we conclude that
	\begin{align*}
	\sum_{x\in f^{-1}(S)} w'(x)\cdot d_p^z(x,C) & \in && (1\pm \eps) \cdot \sum_{x\in S} w(x)\cdot d_p^z(x^o,C) & (\text{Ineq.~\eqref{eq:cor1} and Defn. of $w'$}) \\
	& \in && (1\pm 3\eps)\cdot n \cdot d_p^z(x^o,C) & (\text{Ineq.~\eqref{eq:total_weight}}) \\
	& \in && (1\pm 4\eps)\cdot \cost_{p,z}(X,C). & (\text{Ineq.~\eqref{eq:cor1}})
	\end{align*}

	\paragraph{Case 2: $C\cap H\neq \emptyset$.} For any $x\in X$, if the closest center to $x$ in $C$ is within $C\cap H$, then we have 
	\[
	d_p^z(x, C\cap H) = d_p^z(x,C).
	\]
	Otherwise, letting $c(x)$ be the closest center to $x$ in $C\cap H$, there must exist point $x^o\in X$ such that 
	\[
	d_p(x^o, c(x)) \leq \frac{20z\cdot \diam_p(X)}{\eps}
	\]
	by the definition of $H$.
	Then we have
	\begin{align*}
	d_p^z(x, C\cap H) &\leq && \left( d_p(x, x^o) + d_p(x^o, c(x)) \right)^z & (\text{triangle ineq.}) \\
	& \leq && \left(\frac{20z\cdot \diam_p(X)}{\eps} + \diam_p(X) \right)^z & (d_p(x^o, c(x)) \leq \frac{20z\cdot \diam_p(X)}{\eps}) \\
	& \leq && (1+\eps)\cdot d_p^z(x,C). & (d_p(x, c(x)) \geq \frac{20z\cdot \diam_p(X)}{\eps}) 
	\end{align*}
	Overall, we have that
	\begin{align}
	\label{eq:cor2}
	\cost_{p,z}(X,C)\leq \cost_{p,z}(X,C\cap H)\leq (1+\eps)\cdot \cost_{p,z}(X,C),
	\end{align}
	and
	\begin{align}
	\label{eq:cor3}
	\sum_{x\in f^{-1}(S)} w'(x)\cdot d_p^z(x, C)\leq \sum_{x\in f^{-1}(S)} w'(x)\cdot d_p^z(x, C\cap H)\leq (1+\eps)\cdot \sum_{x\in f^{-1}(S)} w'(x)\cdot d_p^z(x, C).
	\end{align}
	Next, for each $c\in C$, let $c'$ denote its closest point in $H_\tau$. 
	Let $C'$ be the collection of these points $c'$.
	For each $x\in X$,  we have that if	$d_p^z(x,C\cap H)> \frac{\eps \OPT_{p,z}}{n}$,
	then 
	\begin{align*}
	d_p^z(x,C') &\in && \left(d_p(x,C\cap H)\pm \tau\right)^z  & (\text{triangle ineq. and Defn. of $C'$}) \\
	& \in && (1\pm \eps)\cdot d_p^z(x,C\cap H). & (\tau \leq \frac{\eps}{10z}\cdot d_p(x,C\cap H)) 
	\end{align*}
	Thus, we conclude that
	\begin{align}
	\label{eq:net}
	d_p^z(x,C\cap H)  \in (1\pm 2\eps)\cdot d_p^z(x,C') \pm \frac{\eps \OPT_{p,z}}{n}.
	\end{align}
	This implies that
	\begin{eqnarray}
	\label{eq:cor4}
	\begin{split}
	\cost_{p,z}(X,C) &\in && (1\pm \eps)\cdot \cost_{p,z}(X, C\cap H) &  (\text{Ineq.~\eqref{eq:cor2}})\\
	& \in && (1\pm 3\eps)\cdot \cost_{p,z}(X, C')\pm \eps \OPT_{p,z} & (\text{Ineq.~\eqref{eq:net}}) \\
	& \in && (1\pm 4\eps)\cdot \cost_{p,z}(X, C') & (\text{Defn. of $\OPT_{p,z}$})\\
	& \in && (1\pm 4\eps)\cdot \cost_{2,2z}(f(X),f(C')). &  (\text{Eq.~\eqref{eq:cor_embed}})
	\end{split}
	\end{eqnarray}
	Moreover, we have
	\begin{eqnarray}
	\label{eq:cor5}
	\begin{split}
	& && \sum_{x\in f^{-1}(S)} w'(x)\cdot d_p^z(x,C) &\\
	&\in && (1\pm \eps)\cdot \sum_{x\in f^{-1}(S)} w'(x)\cdot d_p^z(x, C\cap H) & (\text{Ineq.~\eqref{eq:cor2}}) \\
	& \in && (1\pm 3\eps)\cdot \sum_{x\in f^{-1}(S)} w'(x)\cdot d_p^z(X, C')\pm \frac{\sum_{x\in S} w(x) \cdot \eps \OPT_{p,z}}{n} & (\text{Ineq.~\eqref{eq:net}}) \\
	& \in && (1\pm 3\eps)\cdot \sum_{x\in S} w(x)\cdot d_2^{2z}(f(x), f(C')) \pm 2\eps\OPT_{p,z} & (\text{Ineq.~\eqref{eq:total_weight}})\\
	& \in && (1\pm 3\eps)\cdot \sum_{x\in S} w(x)\cdot d_2^{2z}(f(x), f(C')) \pm 2\eps\cdot \cost_{2,2z}(f(X), f(C')) & (\text{Defn. of $\OPT_{p,z}$})\\
	& \in && (1\pm 6\eps)\cdot \cost_{2,2z}(f(X),f(C')). &  (\text{Defn. of $S$})
	\end{split}
	\end{eqnarray}
	Combining with Inequalities~\eqref{eq:cor4} and~\eqref{eq:cor5}, we prove Inequality~\eqref{eq:cor_claim}, which implies that $f^{-1}(S)$ with $w'$ is an $O(\eps)$-coreset for \kzC of $X$ with $\ell_p$-metric in $\R^d$.
	Consequently, the theorem is a direct corollary of Theorem~\ref{thm:coreset}.

	Similar to Theorem~\ref{thm:coreset}, the running time is dominated by finding an $2^{O(z)}$-approximate solution, which is $O\left(ndk+nd \log(n/\delta) + k^2 \log^2 n + \log^2(1/\delta) \log^2 n\right)$~\cite{mettu2004optimal}.
\end{proof}

\begin{remark}
	It is unknown whether we can achieve similar results as Corollary~\ref{cor:coreset_lp} for $p> 2$.
	An interesting case is $p=\infty$.
	However, since all metrics can isometrically embed to $\ell_\infty$~\cite{kahane1993some} and there is a lower bound $\Omega(\log n)$ for the coreset size in general metric spaces~\cite{braverman2019coresets}, it is impossible to remove the dependence of $d$ in the coreset size when $p=\infty$.
	Hence, the left open question is that for constant $p>2$, whether we can efficiently construct coresets for \kzC with $\ell_p$-metrics and of size independent of $d$.
\end{remark}

\section{Size lower bounds}
\label{sec:lower}

In this section, we discuss size lower bounds of coresets for \kzC.
Theorem~\ref{thm:lower} shows that the coreset should have an exponential dependence of size on $z$, which matches Theorem~\ref{thm:coreset}.
For any fixed $k,z$, the idea is to construct a bad instance $X\subseteq \R^{\min\left\{d,z/20\right\}}$ such that any $0.01$-coreset has size $\Omega\left(k\cdot \min\left\{2^{z/20},d\right\}\right)$.

\begin{theorem}[\bf{Restatement of Theorem~\ref{thm:lower}}]
	\label{thm:lower_restated}
	For every $z> 0$ and integers $d,k\geq 1$, there exists a point set $X$ in the Euclidean space $\R^d$ such that any $0.01$-coreset for \kzC over $X$ has size $\Omega\left(k\cdot \min\left\{2^{z/20},d\right\}\right)$.
\end{theorem}

\begin{proof}
	We first consider the case that $d\leq 2^{\lfloor z/20 \rfloor}$.
	If $z\leq 100$, we consider $X$ as a collection of $k$ distinct points.
	Since $\cost_z(X,X)=0$, a $0.01$-coreset for \kzC over $X$ must have size $k=\Omega(2^{d} k)$, which completes the proof.
	Hence, we assume that $z\geq 100$ and $d\leq 2^{\lfloor z/20 \rfloor}$ in the following.

	\noindent \textbf{Case $k$=1.}
	We first consider a simple case that $k=1$.
	Let $X=\left\{e_1,-e_1,\ldots, e_d,-e_d\right\}$ denote the collection of unit coordinate vectors and suppose $S\subseteq \R^d$ is a $0.01$-coreset together with a weight function $w:S\rightarrow \R_{\geq 0}$.
	By contradiction we assume that $|S|=o(d)$.

	\eat{
		We first assume that $S$ is symmetric with respect to the origin.
		If not, we construct $S'= \left\{x,-x\mid x\in S\right\}$ together with a weight function $w':S'\rightarrow \R_{\geq 0}$ satisfying that $w'(x)=w'(-x)=w(x)/2$ for every $x\in S$.
		Then $|S'|=2\cdot|S|$.
		Moreover, for any center $c\in \R^d$, we have
		\begin{eqnarray*}
			\begin{split}
				\sum_{x\in S'} w'(x)\cdot d^z(x,c) &= && \sum_{x\in S} \frac{w(x)}{2}\cdot\left(d^z(x,c)+d^z(-x,c)\right) && (\text{Defn. of $S$}) \\
				& = && \frac{1}{2}\sum_{x\in S} w(x)\cdot d^z(x,c) + \frac{1}{2}\sum_{x\in S} w(x)\cdot d^z(x,c') && (d(-x,c)=d(x,-c)) \\
				&\in && \frac{1}{2}(1\pm 0.01)\cdot \cost_z(X,c) + \frac{1}{2}(1\pm 0.01)\cdot \cost_z(X,-c) && (\text{Coreset $S$}) \\
				&\in &&(1\pm 0.01)\cdot  \cost_z(X,c). && (\text{Symmetry of $X$})
			\end{split}
		\end{eqnarray*}
		Hence, $S'$ is also a 0.01-coreset for \kzC of size $o(2^{z/2} k)$.
	}
	Let $H = \left\{x\in \R^d: \|x\|_2\leq 1\right\}$ denote the unit ball centered at the origin.
	We first have the following claim.
	\begin{claim}
		\label{claim:lower1}
		$\sum_{x\in S\cap H} w(x)\leq 1.02\cdot 2d$.
	\end{claim}
	\begin{proof}
		By contradiction assume that $\sum_{x\in S\cap H} w(x)> 1.02\cdot 2d$.
		Consider center $c=10^5 z\cdot e_1$.
		On one hand, we have
		\begin{eqnarray*}
			\begin{split}
				\sum_{x\in S} w(x)\cdot d^z(x,c)&\geq &&\sum_{x\in S\cap H} w(x)\cdot d^z(x,c)&&\\
				&\geq&& \left(\sum_{x\in S\cap H} w(x)\right)\cdot (10^5 z-1)^z && (\text{Defn. of $H$}) \\
				&\geq && 1.02\cdot 2d\cdot (10^5 z-1)^z. &&
			\end{split}
		\end{eqnarray*}
		On the other hand, we have
		\begin{eqnarray*}
			\begin{split}
				\sum_{x\in S} w(x)\cdot d^z(x,c)&\leq &&1.01\cdot \sum_{x\in X} d^z(x,c)&& (\text{Defn. of $S$}) \\
				&\leq&& 1.01\cdot 2d\cdot (10^5 z+1)^z. && (\text{Construction of $X$}) \\
				&< && 1.02\cdot 2d\cdot (10^5 z-1)^z, &&(z>100)
			\end{split}
		\end{eqnarray*}
		which is a contradiction. 
		This completes the proof.
	\end{proof}
	We also have the following lemma.
	\begin{lemma}
		\label{lm:lower2}
		At least one of the following properties holds:
		\begin{itemize}
			\item $\sum_{x\in S} w(x)\cdot d^z(x,0)> 2\cdot 2d$;
			\item There exists $y\in X$ such that $\sum_{x\in S} w(x)\cdot d^z(x,y)< 0.5\cdot 2^z$.
		\end{itemize}
	\end{lemma}
	\begin{proof}
		We assume that $\sum_{x\in S} w(x)\cdot d^z(x,0)\leq 2\cdot 2d$.
		Otherwise, we have done.

		We partition $\R^d$ into $2d$ cones $P_x s$, where for each $x\in X$, cone $P_x$ is the collection of points in $\R^d$ that is closest to $x$ among $X$.
		Since $|S|=o(d)$, there must exist a cone $P_x$ such that $P_x\cap S=\emptyset$.
		W.l.o.g., we assume that $P_{e_1}$ satisfying that $P_{e_1}\cap S=\emptyset$.
		Then it suffices to prove $\sum_{x\in S} w(x)\cdot d^z(x,-e_1)< 0.5\cdot 2^z$.
		We consider points in $S\cap H$ and $S\setminus H$ separately.

		We first discuss those points in $S\cap H$. 
		Since $P_{e_1}\cap S=\emptyset$, every $x\in S$ satisfies that $x_1\leq \max\left\{|x_2|,\ldots,|x_d|\right\}$.
		Then for any $x\in S\cap H$, we have 
		\begin{eqnarray}
		\begin{split}
		\label{eq:lower1}
		x_1^2&\leq && \frac{x_1^2}{2}+\frac{\max\left\{|x_2|,\ldots,|x_d|\right\}^2}{2} &&\\
		&\leq && \frac{\|x\|_2^2}{2} &&\\
		&\leq && \frac{1}{2}, && (x\in H)
		\end{split}
		\end{eqnarray} 
		i.e., $x_1\leq \frac{1}{\sqrt{2}}$.
		Moreover, for each $x\in S\cap H$, we know that
		\begin{eqnarray}
		\label{eq:lower2}
		\begin{split}
		d^z(x, -e_1) & = && \left((1+x_1)^2+\sum_{i=2}^{d} x_i^2\right)^{z/2} && (\text{by definition}) \\
		& \leq && \left( (1+x_1)^2 + 1-x_1^2 \right)^{z/2} && (x\in H) \\
		& = && \left( 2+2 x_1 \right)^{z/2} && \\
		& \leq && (2+\sqrt{2})^{z/2} && (\text{Ineq.~\eqref{eq:lower1}}) \\
		& \leq &&  2^{0.9z}. &&
		\end{split}
		\end{eqnarray}
		It implies that
		\begin{eqnarray}
		\label{eq:lower3}
		\begin{split}
		\sum_{x\in S\cap H} w(x)\cdot d^z(x,-e_1) & \leq && \left(\sum_{x\in S\cap H} w(x) \right)\cdot 2^{0.9z} && (\text{Ineq.~\eqref{eq:lower2}}) \\
		& \leq && 1.02\cdot 2d\cdot 2^{0.9z} && (\text{Claim~\ref{claim:lower1}}) \\
		& \leq && 2.04\cdot 2^{0.95z}. && (d\leq 2^{\lfloor z/20 \rfloor})
		\end{split}
		\end{eqnarray}

		We then discuss those points in $S\setminus H$.
		For each $x\in S\setminus H$, we have that
		\begin{eqnarray}
		\label{eq:lower4}
		\begin{split}
		\frac{d^z(x,-e_1)}{d^z(x,0)} & = && \frac{\left((1+x_1)^2+\sum_{i=2}^{d} x_i^2\right)^{z/2}}{\left(\sum_{i=1}^{d} x_i^2\right)^{z/2}} && (\text{by definition}) \\
		& = && \left(1+\frac{1+2x_1}{\sum_{i=1}^{d} x_i^2}\right)^{z/2}&& \\
		& \leq && \left(1+\frac{1+\sqrt{2}\cdot \sqrt{\sum_{i=1}^{d} x_i^2}}{\sum_{i=1}^{d} x_i^2}\right)^{z/2} && (x_1\leq \max\left\{|x_2|,\ldots,|x_d|\right\}) \\
		& \leq && \left(1+1+\sqrt{2} \right)^{z/2} && (\sum_{i=1}^{d} x_i^2\geq 1) \\
		& \leq && 2^{0.9z}. &&
		\end{split}
		\end{eqnarray}
		It implies that
		\begin{eqnarray}
		\label{eq:lower5}
		\begin{split}
		\sum_{x\in S\setminus H} w(x)\cdot d^z(x,-e_1) & = && \sum_{x\in S\setminus H} w(x)\cdot d^z(x,0) \cdot \frac{d^z(x,-e_1)}{d^z(x,0)} && \\
		& \leq && 2^{0.9z}\cdot \sum_{x\in S\setminus H} w(x)\cdot d^z(x,0) && (\text{Ineq.~\eqref{eq:lower4}}) \\
		& \leq && 2^{0.9z}\cdot  \sum_{x\in S} w(x)\cdot d^z(x,0) && \\
		& \leq && 2\cdot 2d\cdot 2^{0.9z} && (\text{by assumption}) \\
		& \leq && 4\cdot 2^{0.95z}. && (d\leq 2^{\lfloor z/20 \rfloor})
		\end{split}
		\end{eqnarray}

		Combining Inequalities~\eqref{eq:lower3} and~\eqref{eq:lower5}, we directly conclude that
		\begin{eqnarray*}
			\begin{split}
				\sum_{x\in S} w(x) \cdot d^z(x,-e_1) & = && \sum_{x\in S\cap H} w(x) \cdot d^z(x,-e_1)+\sum_{x\in S\setminus H} w(x) \cdot d^z(x,-e_1) && \\
				& \leq && 2.04\cdot 2^{0.95z}+4\cdot 2^{0.95z} && (\text{Ineqs.~\eqref{eq:lower3} and~\eqref{eq:lower5}}) \\
				& \leq && 0.5\cdot 2^z. && (z\geq 100)
			\end{split}
		\end{eqnarray*}
		We complete the proof.
	\end{proof}
	By Lemma~\ref{lm:lower2}, we conclude that $S$ is not a 0.01-coreset.
	The reason is that if $\sum_{x\in S} w(x)\cdot d^z(x,0)> 2\cdot 2d$ holds, then
	\[
	\sum_{x\in S} w(x)\cdot d^z(x,0) > 2\cdot \sum_{x\in X} d^z(x,0),
	\]
	which is a contradiction.
	Otherwise if there exists $y\in X$ such that $\sum_{x\in S} w(x)\cdot d^z(x,y)< 0.5\cdot 2^z$, then since $-y\in X$ and $d(y,-y)=2$,
	\[
	\sum_{x\in S} w(x)\cdot d^z(x,y)<0.5 \cdot d^z(-y,y)< 0.5\cdot \sum_{x\in X} d^z(x,y),
	\]
	which is also a contradiction. \\
	
	\noindent \textbf{Case $k\geq 1$.} Next, we prove for the general $k$.
	We construct $X$ as a copy of $k$ pieces $X^{(1)},\ldots,X^{(k)}$ of the instance $X$ in Case $k=1$, where each copy is far away from each other, say for any $i\neq j\in [k]$, we let 
	\[
	d(X^{(i)},X^{(j)}) = \min_{x\in X^{(i)}, y\in X^{(j)}}d(x,y)\geq 10^{100 z k}.
	\]
	For each $i\in [k]$, we assume that $o_i$ is the average of $X^{(i)}$, which plays the same role as the origin in Case $k=1$.
	Similarly, for each $i\in [k]$, we denote $H^{(i)}$ to be the unit ball centered at $o_i$.
	Suppose $S\subseteq \R^d$ is a 0.01-coreset together with a weight function $w:S\rightarrow \R_{\geq 0}$.
	By contradiction we assume that $|S|=o(2^{z/20} k)=o(d k)$.

	By a similar argument as in Claim~\ref{claim:lower1}, we first have the following claim.
	\begin{claim}
		\label{claim:lowerk1}
		for each $i\in [k]$, $\sum_{x\in S\cap H^{(i)}} w(x) \leq 1.02\cdot 2d$.
	\end{claim}
	\begin{proof}
		By contradiction and w.l.o.g., we assume that 
		\[
		\sum_{x\in S\cap H^{(1)}} w(x) > 1.02\cdot 2d.
		\]
		We also assume that $o_1=0$.
		Consider a $k$-center set $C=\left\{c_1=10^5 zk\cdot e_1, c_2 = o_2, \ldots, c_k=o_k\right\}$.
		On one hand, we have
		\begin{eqnarray*}
			\begin{split}
				\sum_{x\in S} w(x)\cdot d^z(x,C)&\geq &&\sum_{x\in S\cap H^{(1)}} w(x)\cdot d^z(x,c_1)&&\\
				&\geq&& \left(\sum_{x\in S\cap H^{(1)}} w(x)\right)\cdot (10^5 zk-1)^z && (\text{Defn. of $H^{(1)}$}) \\
				&\geq && 1.02\cdot 2d\cdot (10^5 zk-1)^z. && (\text{by assumption})
			\end{split}
		\end{eqnarray*}
		On the other hand, we have
		\begin{eqnarray*}
			\begin{split}
				\sum_{x\in S} w(x)\cdot d^z(x,C)&\leq &&1.01\cdot \sum_{x\in X} d^z(x,C)&& (\text{Defn. of $S$}) \\
				& = && 1.01\cdot \sum_{i\in [k]} \sum_{x\in X^{(i)}} d^z(x,c_i) && (\text{Construction of $X$}) \\
				&\leq&& 1.01\cdot \left(2d\cdot (10^5 zk+1)^z + (k-1)\cdot 2d \right). && (\text{Defn. of $C$}) \\
				&< && 1.02\cdot 2d\cdot (10^5 zk-1)^z, &&
			\end{split}
		\end{eqnarray*}
		which is a contradiction. 
		This completes the proof.
	\end{proof}
	For each $i\in [k]$, we denote $S^{(i)}$ to be the collection of points in $S$ that is closest to $o_i$ (breaking ties arbitrarily).
	Let $A\subseteq [k]$ denote the collection of $i$ satisfying that $|S^{(i)}|=o(2^{z/20}) = o(d)$.
	Since $|S|=o(2^{z/20} k)$, we have that $|A|\geq 0.99 k$.
	We partition $\R^d$ into $2d k$ cells $P^{(i)}_x s$, where for each $x\in X$, cell $P^{(i)}_x$ is the collection of points in $\R^d$ that is closest to point $x\in X^{(i)}$ among $X$.
	Then for each $i\in A$, there must exist a cell $P^{(i)}_x$ such that $P^{(i)}_x\cap S^{(i)}=\emptyset$. 
	W.l.o.g., we assume this cell to be $P^{(i)}_{o_i+e_1}$.
	Let $C = \left\{c_1=o_1,\ldots, c_k = o_k\right\}$.
	Let $C'$ denote the $k$-center set in which $c'_i = o^{i}$ if $i\notin A$ and $c'_i = o_i-e_1$ if $i\in A$.
	We have the following lemma.
	\begin{lemma}
		\label{lm:lowerk2}
		At least one of the following properties holds:
		\begin{itemize}
			\item $\sum_{x\in S} w(x)\cdot d^z(x,C)> 1.01\cdot 2d k$;
			\item $\sum_{x\in S} w(x)\cdot d^z(x,C')< 0.5\cdot 2^z k$.
		\end{itemize}
	\end{lemma}
	\begin{proof}
		We assume that $\sum_{x\in S} w(x)\cdot d^z(x,C)\leq 1.01\cdot 2d k$.
		Otherwise, we have done.

		For each $i\in [k]$, we denote $\overline{H^{(i)}} = \left(\bigcup_{x\in X^{(i)}} P^{(i)}_x \right) \setminus H^{(i)}$.
		By the construction of $X$ and $P^{(i)}$, we know that for any $x\notin \overline{H^{(i)}}$,
		\begin{align}
		\label{eq:diam}
		d(x, c_i) > 100.
		\end{align}
		Similar to Lemma~\ref{lm:lower2}, we consider points in $S^{(i)}\cap H^{(i)}$ and $S^{(i)}\setminus H^{(i)}$ for $i\in A$ separately.

		We first discuss those points in $S^{(i)}\cap H^{(i)}$ for $i\in A$. 
		By the same argument as in Inequality~\eqref{eq:lower3}, we directly have
		\begin{eqnarray*}
			\begin{split}
				\sum_{x\in S^{(i)}\cap H^{(i)}} w(x)\cdot d^z(x,c'_i) & \leq && 	\sum_{x\in S\cap H^{(i)}} w(x)\cdot d^z(x,-e_1) &&\\
				& \leq && \left(\sum_{x\in S\cap H^{(i)}} w(x) \right)\cdot 2^{0.9z} && (\text{Ineq.~\eqref{eq:lower2}}) \\
				& \leq && 1.02\cdot 2d\cdot 2^{0.9z} && (\text{Claim~\ref{claim:lowerk1}}) \\
				& \leq && 2.04\cdot 2^{0.95z}. && (d\leq 2^{\lfloor z/20 \rfloor})
			\end{split}
		\end{eqnarray*}
		Hence, we conclude that
		\begin{align}
		\label{eq:lowerk1}
		\sum_{i\in A} \sum_{x\in S^{(i)}\cap H^{(i)}} w(x)\cdot d^z(x,c'_i) \leq 2.04\cdot 2^{0.95z} k.
		\end{align}

		Then we discuss those points in $S^{(i)}\setminus H^{(i)}$ for $i\in A$. 
		If $x\in \overline{H^{(i)}}\setminus H^{(i)}$, then by the same argument as in Inequality~\eqref{eq:lower4} we have
		\[
		\frac{d^z(x,c'_i)}{d^z(x,c_i)} \leq 2^{0.9z}.
		\]
		If $x\in S^{(i)}\setminus \overline{H^{(i)}}$, since $d(x,c_i)> 100 $ by Inequality~\eqref{eq:diam}, we have
		\[
		\frac{d^z(x,c'_i)}{d^z(x,c_i)} \leq \frac{\left(d(x,c_i)+1\right)^z}{d^z(x,c_i)}\leq 2^{0.9z}.
		\]
		Combining with the above two inequalities, we have that for any $x\in S^{(i)}\setminus H^{(i)}$ for $i\in A$,
		\begin{align}
		\label{eq:ratio}
		\frac{d^z(x,c'_i)}{d^z(x,c_i)} \leq 2^{0.9z}.
		\end{align}
		Then by a similar argument as in Inequality~\eqref{eq:lower5}, we have
		\begin{eqnarray}
		\label{eq:lowerk2}
		\begin{split}
		\sum_{i\in A}\sum_{x\in S^{(i)}\setminus H^{(i)}} w(x)\cdot d^z(x,c'_i) & = && \sum_{i\in A} \sum_{x\in S^{(i)}\setminus H^{(i)}} w(x)\cdot d^z(x,c_i) \cdot \frac{d^z(x,c'_i)}{d^z(x,c_i)} && \\
		& \leq && 2^{0.9z}\cdot \sum_{i\in A} \sum_{x\in S^{(i)}\setminus H^{(i)}} w(x)\cdot d^z(x,c_i) && (\text{Ineq.~\eqref{eq:ratio}}) \\
		& \leq && 2^{0.9z}\cdot  \sum_{x\in S} w(x)\cdot d^z(x,C) && (\text{Defn. of $S^{(i)}$})\\
		& \leq && 1.01\cdot 2d k\cdot 2^{0.9z} && (\text{by assumption}) \\
		& \leq && 2.02\cdot 2^{0.95z} k. && (d\leq 2^{\lfloor z/20 \rfloor})
		\end{split}
		\end{eqnarray}

		Overall, we have the following inequality.
		\begin{eqnarray*}
			\begin{split}
				& && \sum_{x\in S} w(x) \cdot d^z(x,C') &&\\
				& \leq && \sum_{i\in [k]} \sum_{x\in S^{(i)}} w(x) \cdot d^z(x,c'_i) && \\
				& = && \sum_{i\in A} \sum_{x\in S^{(i)}} w(x) \cdot d^z(x,c'_i)+\sum_{i \notin A} \sum_{x\in S^{(i)}} w(x) \cdot d^z(x,c'_i)&& \\
				& = && \sum_{i\in A} \sum_{x\in S^{(i)}} w(x) \cdot d^z(x,c'_i)+\sum_{i\notin A} \sum_{x\in S^{(i)}} w(x) \cdot d^z(x,c_i)&& (c'_i=c_i, \forall i\notin A)\\
				& \leq && \sum_{i\in A} \sum_{x\in S^{(i)}} w(x) \cdot d^z(x,c'_i)+ 1.01\cdot 2d k&& (\text{by assumption})\\
				& = && \sum_{i\in A} \left(\sum_{x\in S^{(i)}\cap H^{(i)}} w(x) \cdot d^z(x,c'_i) +\sum_{x\in S^{(i)}\setminus H^{(i)}} w(x) \cdot d^z(x,c'_i)\right) &&\\
				& &&+ 2.02\cdot 2^{0.05z} k&& (\text{Defn. of $d$})\\
				& \leq && 2.04\cdot 2^{0.95z} k+2.02\cdot 2^{0.95z} k + 2.02\cdot 2^{0.05z} k&& (\text{Ineqs.~\eqref{eq:lowerk1} and~\eqref{eq:lowerk2}}) \\
				& \leq && 0.5\cdot 2^z. && (z\geq 100)
			\end{split}
		\end{eqnarray*}
		We complete the proof.
	\end{proof}

	By Lemma~\ref{lm:lowerk2}, we conclude that $S$ is not a 0.01-coreset.
	The reason is that if $\sum_{x\in S} w(x)\cdot d^z(x,C)> 1.01\cdot 2d$ holds, then
	\[
	\sum_{x\in S} w(x)\cdot d^z(x,C) > 1.01\cdot \sum_{x\in X} d^z(x,C),
	\]
	which is a contradiction.
	Otherwise if $\sum_{x\in S} w(x)\cdot d^z(x,C')< 0.5\cdot 2^z k$, then since $o_i+e_1\in X$ and $d(o_i+e_1,c'_i)=2$ for each $i\in A$,
	\begin{eqnarray*}
		\begin{split}
			\sum_{x\in S} w(x)\cdot d^z(x,C')&<&& 0.99\cdot \sum_{i\in A} 2^z && (|A|\geq 0.99 k) \\
			&=&& 0.99\cdot \sum_{i\in A} d^z(o_i+e_1,C') && (d(o_i+e_1,c'_i)=2) \\
			&\leq && 0.99\cdot \sum_{x\in X} d^z(x,C'),
		\end{split}
	\end{eqnarray*}
	which is also a contradiction. 
	
	For the case that $d\geq 2^{\lfloor z/20 \rfloor}$, we construct an instance $X=\left\{e_1,-e_1,\ldots, e_{2^{\lfloor z/20 \rfloor}},-e_{2^{\lfloor z/20 \rfloor}}\right\}$. 
	By the same argument as above, we can prove that any $0.01$-coreset that preserves for all \kzC objectives with respect to center sets on the subspace spanned by $X$ should have size $\Omega(2^{\lfloor z/20 \rfloor} k)$, which completes the proof. 
\end{proof}

\section{A failed attempt: projecting $X$ using Johnson–Lindenstrauss}
\label{sec:failed}

Our first attempt is motivated by a recent paper~\cite{makarychev2019performance}, which states that by a Johnson–Lindenstrauss transform $g: \R^d\rightarrow \R^m$ where $m=O\left(\eps^{-2}\log (k/\eps)\right)$, the optimal \kMedian objective with respect to each ``partition'' of $X$ can be approximately preserved, i.e., for any partition $\left(X_1,\ldots,X_k\right)$ of $X$, 
\begin{align}
\label{eq:JL}
\sum_{i=1}^{k} \min_{c\in \R^m} \sum_{x\in X_i} d(g(x), c) \in (1\pm \eps)\cdot \sum_{i=1}^{k} \min_{c\in \R^d} \sum_{x\in X_i} d(x, c).
\end{align}
%
%
Recall that $\calC$ denotes the collection of all ordered subsets (repeats allowed) of $\R^d$ of size $k$ ($k$-center sets).
We want to generalize~\cite{makarychev2019performance} such that all \kMedian objectives in $\calC$ are approximately preserved. 
Concretely speaking, we want to construct (randomized) mappings $f:\R^d\rightarrow \R^m$ and $\phi: (\R^d)^k \rightarrow (\R^m)^k$ 
(independent of the choice of $X$) such that for any $k$-center set $C\in\calC$,
\begin{align}
\label{eq:JL1}
\sum_{x\in X} d\left(f(x),\phi(C)\right)\in (1\pm \eps)\cdot \cost_1(X,C),
\end{align}
where $\cost_1(X,C) = \sum_{x\in X} d(x,C)$ by Equation~\eqref{eq:DefCost}.
Denote $f(A):=\left\{f(x): x\in A\right\}$ for any point set $A\subset \R^d$.
By the \FL framework (Theorem~\ref{thm:fl11_bfl16}), there exists a weighted subset $S\subseteq X$ together with a weight function $u:D\rightarrow \R_{\geq 0}$ and of size $\poly(k/\eps)$, such that $f(S)$ is an $\eps$-coreset for \kMedian over $f(X)$, i.e., for any $k$-center set $C\subset \R^m$,
\begin{align}
\label{eq:JL_coreset}
\sum_{x\in S} w(x)\cdot d(f(x),C)\in (1\pm \eps)\cdot \sum_{x\in X} d(f(x), C).
\end{align}
Then, as in Inequality~\eqref{eq:JL1}, we may argue that for any $k$-center set $C\in \calC$,
\begin{align}
\label{eq:JL2}
\sum_{x\in S} w(x)\cdot d\left(f(x),\phi(C)\right)\in (1\pm \eps)\cdot \sum_{x\in S} w(x)\cdot d(x,C),
\end{align}
%
%
since mappings $f,\phi$ are independent on the choice of $X$.
Combining Inequalities~\eqref{eq:JL1}-\eqref{eq:JL2}, we have that for any $k$-center set $C\in \calC$, 
\begin{eqnarray*}
	\begin{split}
		\sum_{x\in S} w(x)\cdot d(x,C) &\approx  && \sum_{x\in S} w(x)\cdot d\left(f(x),\phi(C)\right) && (\text{Ineq.~\eqref{eq:JL2}}) \\
		& \approx && \sum_{x\in X} d(f(x), \phi(C)) && (\text{Ineq.~\eqref{eq:JL_coreset}}) \\
		&\approx  && \cost_1(X,C). && (\text{Ineq.~\eqref{eq:JL1}})
	\end{split}
\end{eqnarray*}
This indicates that $S$ is an $O(\eps)$-coreset for \kMedian over $X$.
One attempt of constructing $f,\phi$ is letting $f=g$ be a Johnson–Lindenstrauss transform and $\phi = g^k$ be defined by $\phi(C)=\left\{g(c): c\in C\right\}$.
It is not hard to see that Inequality~\eqref{eq:JL1} does not hold by the following example.
Assume that $X$ consists of one point, then since $f=\phi = g$, there always exists a point $y\in \R^d$ such that 
\[
d(f(x),\phi(y))> (1+\eps)\cdot d(x,y),
\]
which violates Inequality~\eqref{eq:JL1} by letting $C=y$.

Another attempt is to let $f=g$ be a Johnson–Lindenstrauss transform but to construct a different mapping $\phi: (\R^d)^k \rightarrow (\R^m)^k$ such that Inequality~\eqref{eq:JL1} holds.
\cite[Theorem 4.3]{makarychev2019performance} shows that for any point $c\in \R^d$, there exists a point $c'\in \R^m$ such that $d(x,c)\approx d(f(x),c')$ holds for almost all $x\in X$.
Based on this fact, we can show the existence of $k$-center sets $C_1,C_2\subseteq \R^m$ for any $k$-center set $C\in \calC$ such that
\[
\sum_{x\in X} d\left(f(x),C_1\right)\in (1\pm \eps)\cdot \cost_z(X,C),
\] 
and
\[
\sum_{x\in S} w(x)\cdot d\left(f(x),C_2\right)\in (1\pm \eps)\cdot \sum_{x\in S} w(x)\cdot d(x,C).
\] 
If $C_1=C_2$ always hold, we are done by letting $\phi(C)=C_1=C_2$.
However, we do not know how to construct $C_1$ ($C_2$) explicitly, i.e., we only know the existence of $C'_1$ such that 
\[
\sum_{x\in X} d(f(x),C'_1)\leq (1+ \eps)\cdot \cost_z(X,C)
\] 
and $C''_1$ such that 
\[
\sum_{x\in X} d(f(x),C''_1)\geq (1- \eps)\cdot \cost_z(X,C).
\]
It is unknown whether there exists a mapping $\phi$ satisfying Inequality~\eqref{eq:JL1} and~\eqref{eq:JL2} simultaneously.

\section{Proof of~\cite[Theorem 15.6]{feldman2011unified}}
\label{sec:fl11}

The proof of Theorem 15.6 in~\cite{feldman2011unified} has some typos for proving (85) by (84), where (84) does not satisfy the condition of ~\cite[Lemma 14.2]{feldman2011unified}.
To fix the typo, it suffices to prove the following lemma.

\begin{lemma}
	\label{lm:proof_fl11}
	Let $a,b,c\geq 0$, $\eps>0$ and $z\geq 1$.
	If $|a-b|\leq c$ and $|a^z-b^z|>\frac{z c^z}{\eps^{z-1}}$, then we have $|a^z-b^z|\leq z\eps\cdot (\max\left\{a,b\right\})^z$.
\end{lemma} 

\begin{proof}
	Without loss of generality, assume that $a>b>0$.
	By scalability, we can also assume that $b=1$.
	Then we have 
	\[
	a^z-1 > \frac{z c^z}{ \eps^{z-1}} \geq \frac{z (a-1)^z}{\eps^{z-1}}.
	\]
	Moreover, we claim that
	\[
	a^z-1\leq (a-1)z\cdot a^{z-1}.
	\]
	This is because that considering function $f(a)=(a-1)z\cdot a^{z-1}-(a^z-1)$, we have $\nabla_a f(a) = (z-1)z (a^{z-1}-a^{z-2})\geq 0$ when $a\geq 1$ and, hence, $f(a)\geq f(1)=0$.
	Combining the above inequalities, we have $\frac{a-1}{a}\leq \eps$.
	Then we have
	\[
	a^z-1\leq (a-1)z\cdot a^{z-1} \stackrel{\frac{a-1}{a}<\eps}{\leq} z\eps\cdot a^z,
	\]
	which completes the proof.
\end{proof}

\noindent
Let $a = \mathrm{dist}(p,x), b = \mathrm{dist}(p',x)$, $c = \mathrm{dist}(p,p')$ and $\eps'=\eps/z$ in Lemma~\ref{lm:proof_fl11}, we complete the proof of Theorem 15.6 in~\cite{feldman2011unified} from (84) to (85).

\section*{Acknowledgments}
This research was supported in part by NSF CCF-1908347 grant. 

\bibliography{references}
\bibliographystyle{plain}

\end{document}